%% file: ICCAD13_TPLPlace.tex
\documentclass[9pt]{sig-alternate-preprint}
\input{lib/acm_setting}

\graphicspath{{./figs/}}

\begin{document}

\title{
Methodology for Standard Cell Compliance and Detailed Placement for Triple Patterning Lithography
\vspace{.1in}
}

\numberofauthors{1}
\author{
\alignauthor Bei Yu,\ \  Xiaoqing Xu,\ \  Jhih-Rong Gao,\ \  David Z. Pan\\
\affaddr{
ECE Department, University of Texas at Austin, TX, USA\\ 
} 
\affaddr{\{bei, xiaoqingxu, jrgao, dpan\}@cerc.utexas.edu}
}

\maketitle
\thispagestyle{empty}   %essential for ACM template

\input{doc/abstract}

\input{doc/intro}
\input{doc/overview}

\input{doc/cell}

\input{doc/coloring}
\input{doc/singlerow}

\input{doc/dplace}

\input{doc/result}

\input{doc/conclu}

{
\scriptsize
\bibliographystyle{IEEEtran}
\vspace{.1in}
\bibliography{./Ref/Bei,./Ref/MPL,./Ref/Algorithm,./Ref/Lith,./Ref/Place}
}

\end{document}

%% file: lib/acm_setting.tex
\usepackage{subfigure}
\usepackage{graphicx}
\usepackage{cite}
\usepackage{threeparttable}
\usepackage{ntheorem}
\usepackage{color}
\usepackage{multirow}
\usepackage{amsmath}
\usepackage{comment}

\usepackage{algpseudocode}
\usepackage{algorithm}
\theoremstyle{plain}
\theoremheaderfont{\normalfont\bfseries}\theorembodyfont{\slshape}
\theoremsymbol{\ensuremath{\diamondsuit}} \theoremseparator{:}

\theoremstyle{plain}
\theoremsymbol{\ensuremath{\clubsuit}}
\theoremseparator{.}
\newtheorem{problem}{Problem}

\theoremsymbol{\ensuremath{\heartsuit}}
\newtheorem{mylemma}{Lemma}

\theoremstyle{change}
\theorembodyfont{\upshape}
\theoremsymbol{\ensuremath{\ast}}
\theoremseparator{}

\theoremheaderfont{\sc}\theorembodyfont{\upshape}
\theoremstyle{nonumberplain} \theoremseparator{}
\theoremsymbol{\rule{1ex}{1ex}}

\theoremstyle{plain} \theoremsymbol{\ensuremath{\clubsuit}}
\theoremseparator{.}
\newtheorem{mydefinition}{Definition}

%% file: doc/abstract.tex
%\begin{abstract}
%As the feature size of semiconductor process technology nodes further scales to sub-16nm, triple patterning lithography (TPL) is one of the most promising lithography candidates, along with extreme ultra violet (EUV) and electron beam lithography (EBL).
%M1 and contact layers, which are usually deployed with standard cells, are most critical part for modern design.
%Considering TPL during placement stage explores a larger solution space and can further improve the layout decomposability.
%In this paper, we propose the first triple patterning aware detailed placement in 16nm node.
%Besides, timing variations, especially on the critical paths, are considered.
%Experimental results show that, applying TPL friendly detailed placement, the conflicts can be resolved much more easily after the layout decomposition.
%\end{abstract}

\begin{abstract}
As the feature size of semiconductor process further scales to sub-16nm technology node,
triple patterning lithography (TPL) has been regarded one of the most promising lithography candidates.
%along with extreme ultra violet (EUV) and electron beam lithography (EBL).
M1 and contact layers, which are usually deployed within standard cells, are most critical and complex parts for modern digital designs.
%Traditionally TPL conflicts from these two complex layers are solved through a late stage called layout decomposition.
%However, ignoring TPL constraints in early design flow while relying on late design stage may limit the potential to resolve all the conflicts.
Traditional design flow that ignores TPL in early stages may limit the potential to resolve all the TPL conflicts.
In this paper, we propose a coherent framework, including standard cell compliance and detailed placement to enable TPL friendly design.
%Considering TPL constraints during early design stages, such as standard cell compliance, explores a larger solution space and can further improve the layout decomposability.
Considering TPL constraints during early design stages, such as standard cell compliance, improves the layout decomposability.
%which also explores a larger solution space.
With the pre-coloring solutions of standard cells, we present a TPL aware detailed placement, where the layout decomposition and placement can be resolved simultaneously.
Our experimental results show that, with negligible impact on critical path delay, our framework can resolve the conflicts much more easily,
compared with the traditional physical design flow and followed layout decomposition.
\end{abstract}

%% file: doc/intro.tex
%\vspace{-.1in}
\section{Introduction}

As the feature size of semiconductor process technology nodes further scales to sub-16nm,
triple patterning lithography (TPL) is regarded as one of the most promising lithography candidates,
along with extreme ultra-violet lithography (EUVL), directed self-assembly (DSA), and electron beam lithography (EBL) \cite{ITRS,LITH_ICCAD2012_Yu}.
TPL is a natural extension along the paradigm of double patterning lithography (DPL), which has been pushed to its limit in sub-16nm,
to introduce better printability \cite{TPL_SPIE2012_Lucas}.

To deploy TPL process, \textit{layout decomposition} is usually applied to divide the initial layout into three masks.
Then each mask is implemented through one exposure-etch process, through which the layout can be produced.
In initial layout, two features with distance less than minimum coloring distance $d_{min}$ should be assigned into different masks.
One \textit{conflict} occurs when two features whose spacing is less than $d_{min}$.
Sometimes the conflict can be also resolved by inserting \textit{stitch} to split a feature into two touching parts.

TPL layout decomposition problem with conflict and stitch minimization has been well studied in the past few years 
\cite{TPL_SPIE08_Cork,TPL_SPIE2011_Ghaida,TPL_ICCAD2011_Yu,TPL_DAC2012_Fang,TPL_ICCAD2012_Tian,TPLEC_SPIE2013_Yu,TPL_DAC2013_Kuang,TPL_ICCAD2013_Yu}.
%However, these existing works all suffer from one or more of the following issues.
However, most existing work suffers from one or more of the following drawbacks.
(1) Because TPL layout decomposition problem is NP-hard \cite{TPL_ICCAD2011_Yu}, most of the decomposers are based on approximation or heuristic methods,
thus some extra conflicts may be reported \cite{TPL_ICCAD2012_Tian}.
(2) For each design, since the library only contains fixed number of standard cells, layout decomposition would contain lots of redundant works.
For example, if one cell is applied hundreds of times in a single design, it would be decomposed hundreds of times during layout decomposition.
(3) Successfully carrying out these decomposition techniques requires the input layouts to be TPL-friendly.
However, since all these decomposition techniques are applied at post-place/route stage, where all the design patterns are already fixed,
they lack the abilities to resolve some native TPL conflict patterns, e.g., four-clique conflicts.
%If we further analyze those conflicts from decomposition results,
%We observe that conflicts originate from the non-three colorable layout patterns, which can only be removed during early layout design stage.
%In other words, the input layouts should contain no native conflicts, or hard-to-decompose patterns.

\begin{figure}[tb]
  \centering
  \subfigure[] {\includegraphics[width=0.10\textwidth]{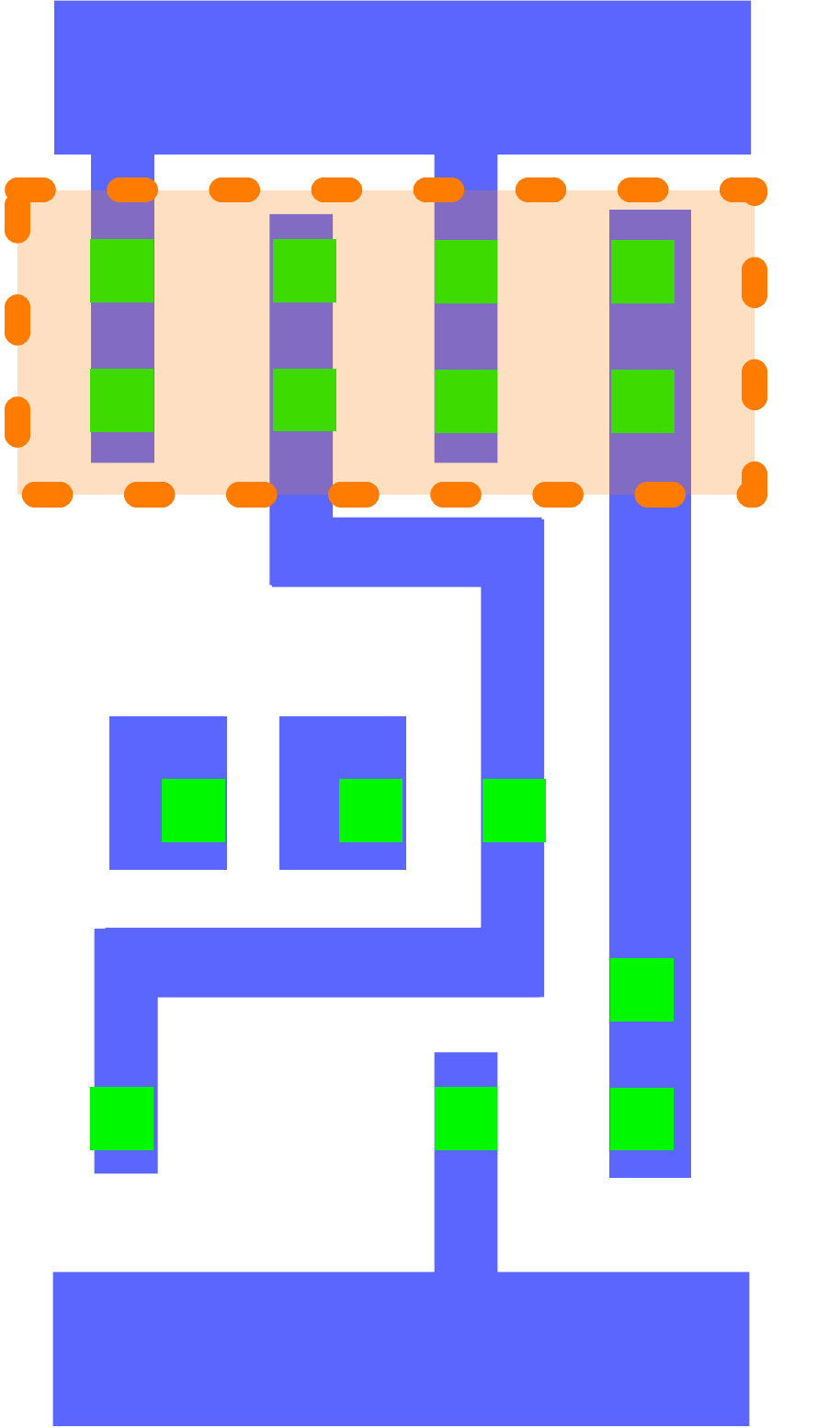}}
  \hspace{.2in}
  \subfigure[] {\includegraphics[width=0.25\textwidth]{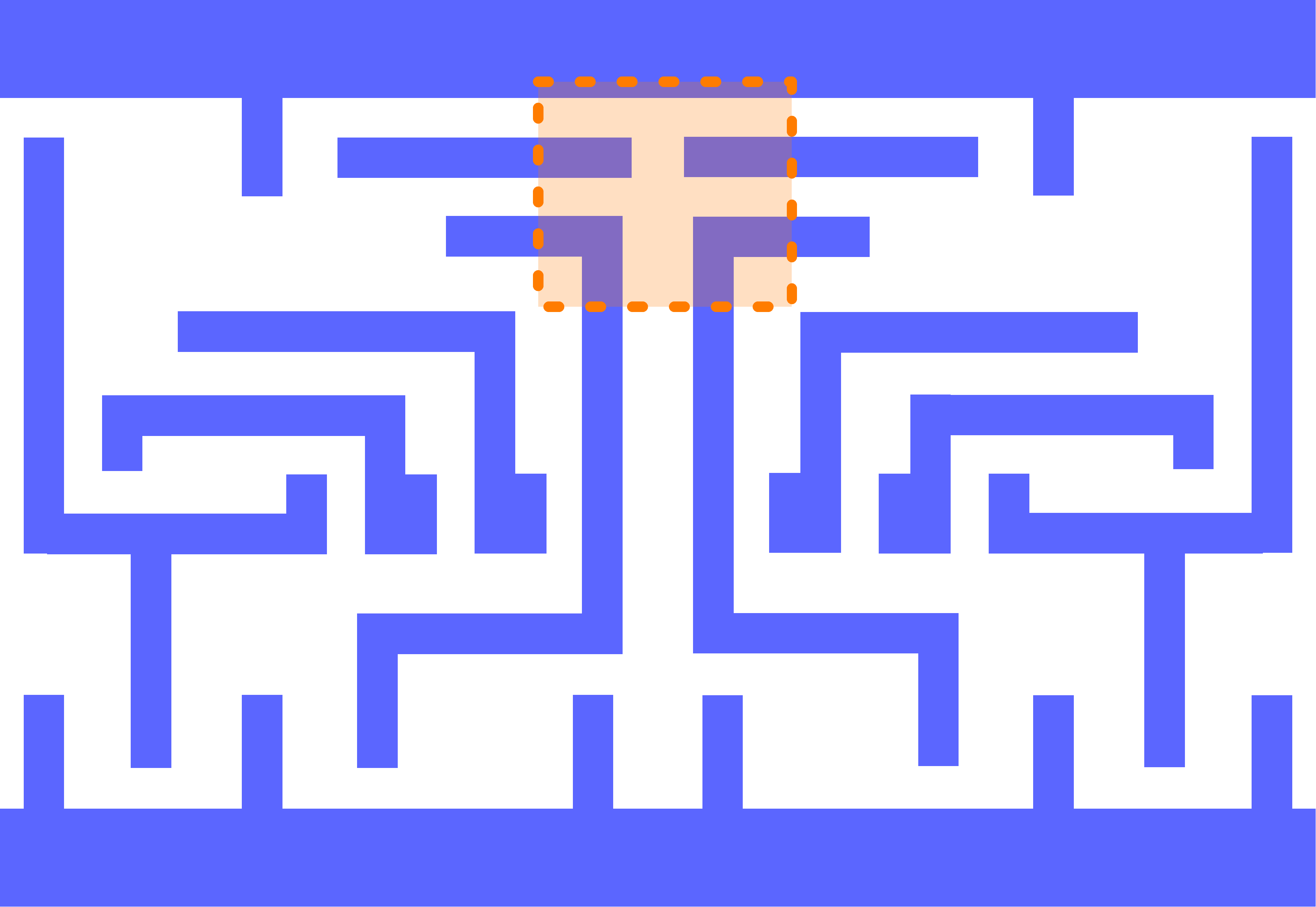}}
  \caption{~Native conflicts from (a) contact layer within a standard cell;~(b) M1 layer between adjacent standard cells.}
  %\caption{~Techniques for removing conflicts on cell boundary. Conflicts are illustrated in red boxes. (a) Flip the cell. (b) Shift the cell. Note: the cell layout is used for demonstration only}
  \label{fig:native_conflict}
  \vspace{-.1in}
\end{figure}

It is observed that the most hard-to-decompose patterns originate from contact and M1 layers.
Fig. \ref{fig:native_conflict} shows two common native TPL conflicts in contact layer and M1 layer, respectively.
% contact layer
As shown in Fig. \ref{fig:native_conflict}(a), contact layout within the standard cell may generate some 4-clique patterns, which is indecomposable.
%cell layout may not be TPL friendly, which leads to unavoidable conflicts for TPL.
Meanwhile, if placement techniques are not TPL friendly, some boundary metals may introduce native conflicts (see Fig. \ref{fig:native_conflict}(b)).
%Metal One is the most complex and dense layout layer, therefore should be carefully designed to avoid the unnecessary conflict or stitch during decomposition stage.
%Sometimes, layout decomposition can point out several conflicts, then designers should manually modify the layout/designs.
%In addition, library...
Since redesigning indecomposable patterns in the final layout requires high ECO efforts,
generating TPL-friendly layouts, especially in the early design stage, becomes urgent and pivotal.
Through these two examples, we can see that both TPL aware standard library design and TPL aware placement are necessary to avoid such indecomposable patterns in final layout.
%A layout configuration without considering TPL at these design stages can make the layout hard to be decomposed.
%So, it may be necessary for us to take TPL into consideration in the earlier physical design stages.
%In other words, we should design TPL friendly layout and make layout free of conflicts in different levels of design.

Liebmann et al. in \cite{DPL_SPIE2011_Liebmann} proposed some guidelines to enable DPL friendly standard cell design and placement.
Besides, there exist several placement studies toward different manufacturing process targets \cite{DFM_ISPD07_Hu,DFM_ICCAD09_Gupta,DFM_SPIE2013_Gao}.
%Actually, there has been related work on double patterning lithography, such as decomposition-aware standard cell design and placement.
%However, little attention has been paid to the TPL issues in earlier design stages.
%Therefore, it is necessary for us to build the framework for TPL-aware physical design flow. 
Recently \cite{DFM_DAC2012_Ma,DFM_ICCAD2012_Lin} proposed TPL aware detailed routing schemes.
%However, none of them handle the special design challenge for triple patterning lithography.
However, to our best knowledge, no previous work has addressed TPL compliance at standard cell or placement level.

%In this paper, we present a systematic study to seamlessly integrate TPL constraints in all design stages of our framework,
%ranging from standard cell conflict removal, standard cell pre-coloring, to detailed placement.
In this paper, we present a systematic framework to seamlessly integrate TPL constraints in early design stages,
comprehending standard cell conflict removal, standard cell pre-coloring and detailed placement together.
Note that our framework is layout decomposition free, that is, the TPL aware detailed placement can generate optimized positions and color assignment solutions for all cells.
Our main contributions are summarized as follows:
\begin{itemize}
  \item We propose systematic standard cell compliance techniques for TPL and coloring solution generation.
  \item We study the standard cell pre-coloring problem, and propose effective methods.
  \item We present the first systematic study for the TPL aware ordered single row placement, where cell placement and color assignment can be solved simultaneously.
  \item Our framework seamlessly integrate decomposition in each key step, therefore no additional layout decomposition is required.
  %\item Our proposed framework is decomposition free, that is, our final detailed placement can directly generate decomposed layouts that guarantee no conflict.
  \item Experimental results show that our framework can achieve zero conflict, meanwhile can effectively reduce the stitch number.
\end{itemize}

The rest of the paper is organized as follows:
Section \ref{sec:overview} provides preliminaries and overview of our methodologies.
Section \ref{sec:cell} proposes standard cell modification to enable TPL friendly cell layout, with negligible timing impact. 
In Section \ref{sec:coloring} the pre-coloring techniques for each cell are proposed, followed by look-up table construction.
Section \ref{sec:singlerow} and Section \ref{sec:dplace} give details on our TPL aware detailed placement.
Section \ref{sec:result} presents the experiment results, followed by conclusion in Section \ref{sec:conclu}.

%% file: doc/overview.tex
%\vspace{-.1in}
\section{Preliminaries}
\label{sec:overview}

%\vspace{-.1in}
%\subsection{Row Based Placement and Color Assignment}
\subsection{Row Structure Layout}

Our framework assumes a row-structure layout,
%where cells in the library are with the same height,
where cells in each row are with the same height,
and power/ground rails are going from the very left to the very right (see Fig. \ref{fig:row_layout}(a)).
Similar assumption was applied in row based TPL layout decomposition \cite{TPL_ICCAD2012_Tian} as well.
The minimum width of metal feature and the minimum spacing between neighboring metal features are denoted as $w_{min}$ and $s_{min}$, respectively.
Besides, we define the minimum spacing between metal features among different rows to be $d_{row}$.
If we further analyze layout patterns in the library, it can be observed the width of power/ground rail is twice the width of metal wire within standard cells \cite{nangate}.
Under the row structure layout, we have the following lemma.
%discuss the correctness of this assumption.

\begin{figure}
  \centering
  \includegraphics[width=0.46\textwidth]{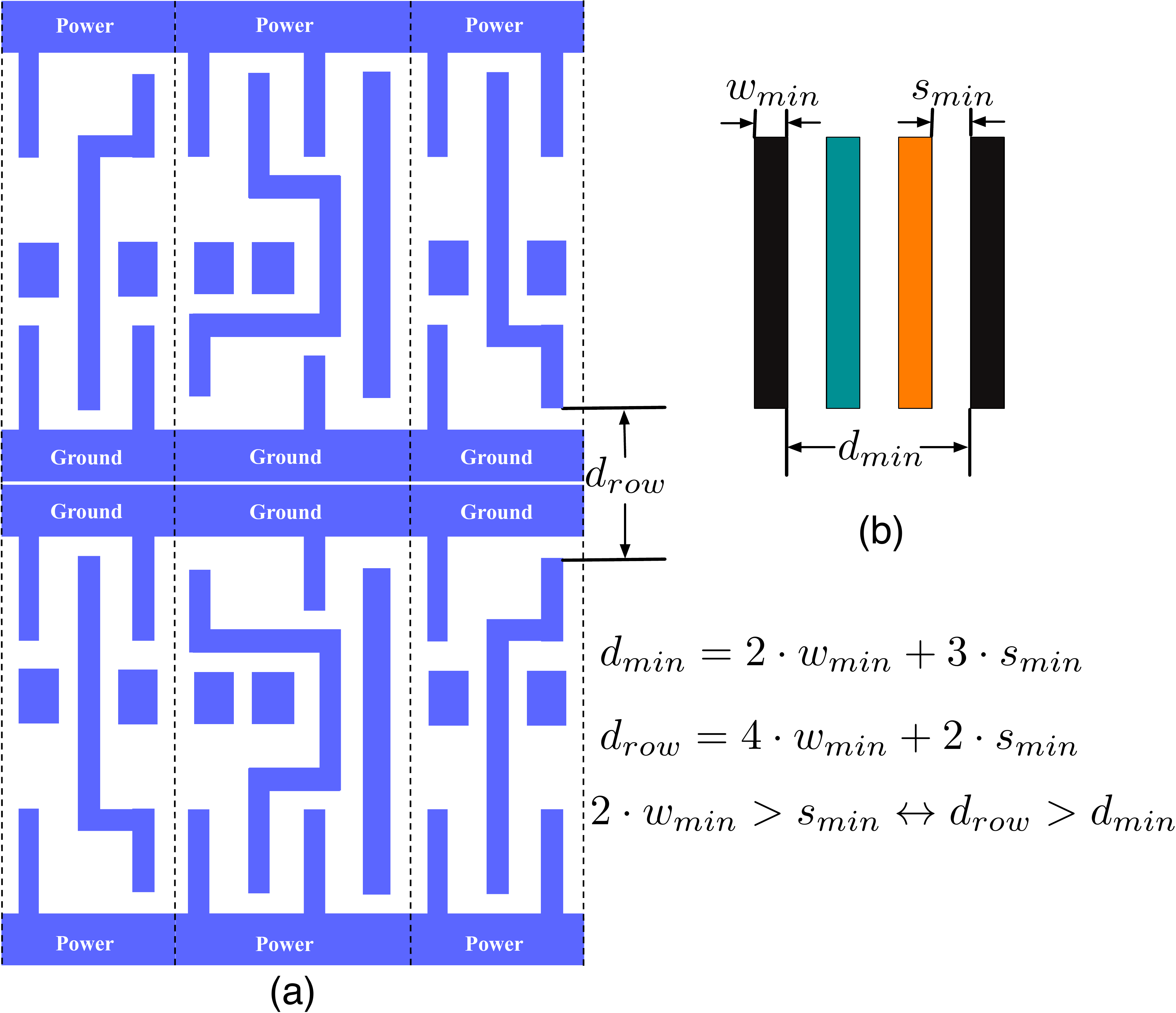}
  \caption{~(a) Minimum spacing between M1 wires among different rows. (b) Minimum spacing between M1 wires with the same color.}
  \label{fig:row_layout}
  \vspace{-.1in}
\end{figure}

\begin{mylemma}
There is no coloring conflict between two M1 wires or contacts that are from different rows.
%\vspace{-.1in}
\end{mylemma}

\begin{proof}
%$W_{min}$ is defined as the minimum width of metal track and $S_{min}$ is the minimum spacing between neighboring metal features. 
For TPL, the layout will be decomposed into three masks, which means layout features within minimum coloring distance will be assigned three colors to increase the pitch between neighboring features.
Then, we can see from the Fig. \ref{fig:row_layout}, the minimum spacing between M1 features with the same color in TPL is
%$d_{min} = 3*W_{min} +3*S_{min}$.
$d_{min} = 2 \cdot w_{min} + 3 \cdot s_{min}$.
We assume the worst case for $d_{row}$, which means the standard cell rows are placed as mirrored cells and allow for no routing channel. 
Thus, $d_{row} = 4 \cdot w_{min} + 2 \cdot s_{min}$.
We should have $d_{row} >  d_{min}$, which equals $2 \cdot w_{min} > s_{min}$. 
This condition can easily be satisfied for M1 layer.
%, therefore the lemma is proved for metal one layer.
% which demonstrates there will be no coloring conflicts for metal features among different rows. 
For the same reason, we can achieve similar conclusion for the contact layer.
\end{proof}
 
Based on the row-structure assumption, the whole layout can be divided into rows, and layout decomposition or coloring assignment can be carried out for each row.
Without loss of generality, for each row the power/ground rails are assigned the color 1 (\textit{default color}).
Then the decomposed results for each row will not induce coloring conflicts among different rows.
In other words, the coloring assignment results for each row can be merged together, without losing optimality.

%\vspace{-.1in}
\subsection{Overall Design Flow}

\begin{figure}[h]
  \centering
  \includegraphics[width=0.40\textwidth]{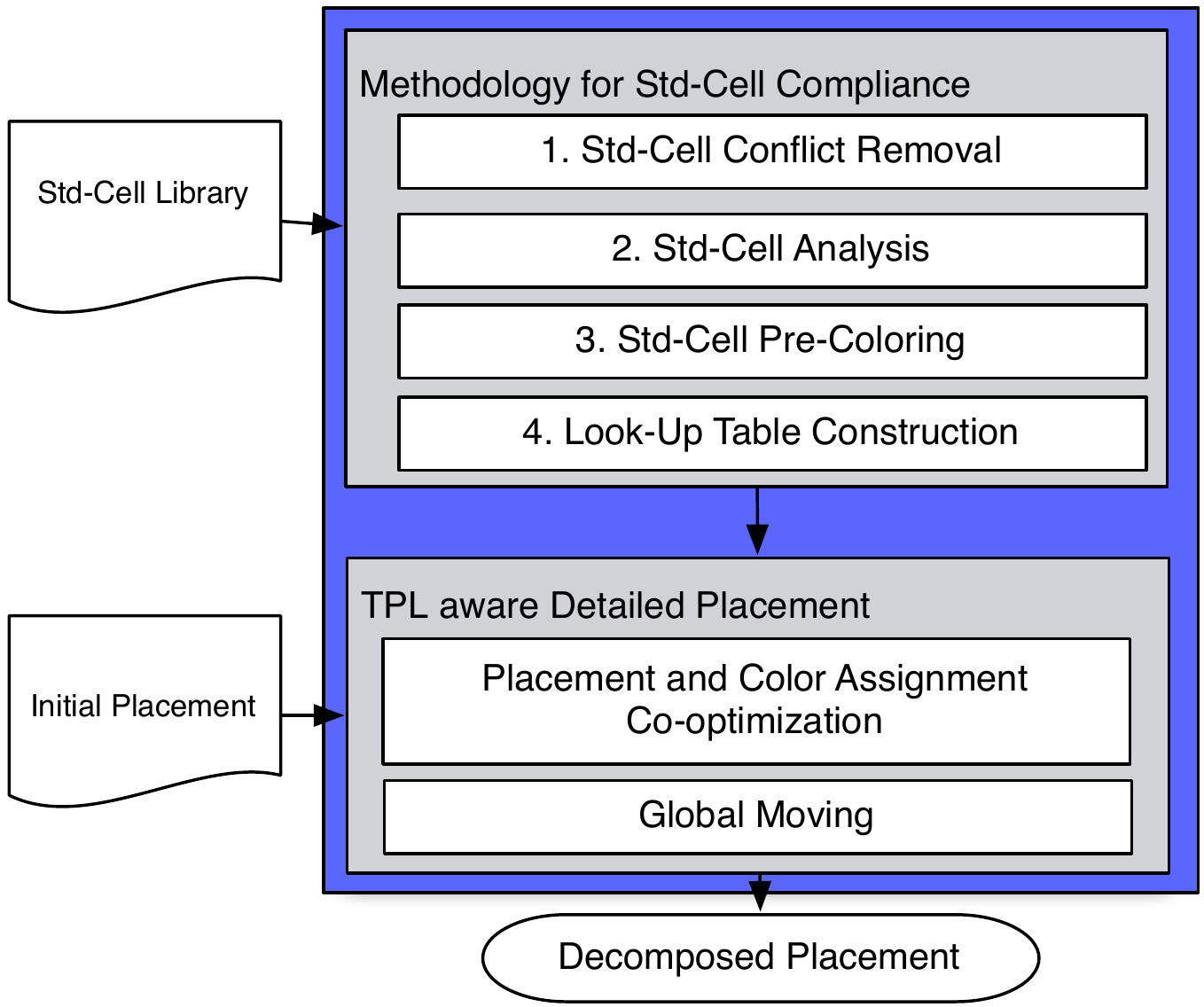}
  \caption{~Overall flow of the methodologies for standard cell compliance and detailed placement.}
  \label{fig:overflow}
  \vspace{-.1in}
\end{figure}

The overall flow of our proposed framework is illustrated in Fig. \ref{fig:overflow}.
It consists of two stages: methodologies for standard cell compliance, and TPL aware detailed placement.
The standard cell compliance techniques include standard cell conflict removal, timing analysis, standard cell pre-coloring, and lookup table generation.
%The overall design flow ranges from standard cell layout modification and timing analysis, standard cell pre-coloring, to detailed placement 
The standard cell compliance techniques ensure that, for each cell, TPL friendly cell layout and a set of pre-coloring solutions will be provided.

Note that since triple patterning lithography constraints are seamlessly integrated into our coherent design flow,
we do not need a separate additional step of layout decomposition.
In other words, the output of our framework is decomposed layouts that have resolved cell placement and color assignment simultaneously.

%% file: doc/cell.tex
%\vspace{-.1in}
\section{Standard Cell Compliance}
\label{sec:cell}

It is observed that without considering TPL in standard cell design, the cell library may involve several cells with native TPL conflict (see Fig. \ref{fig:native_conflict} (a) for one example).
The inner native TPL conflict cannot be resolved through either cell shift or layout decomposition.
Since one cell may be applied many times in one single design, such inner native conflict may cause hundreds of coloring conflicts in final layout.
To achieve TPL friendly layout after the physical design flow, we should first ensure the standard cell layout compliance for TPL.
Specifically, we will manually remove all 4-clique conflicts through standard cell modification.
%in the Nangate standard cell library.
Then, parasitic extraction and SPICE simulation are applied to analyze the timing impact for the cell modification.
%we expect for negligible timing impact from layout modification.

\begin{figure}[tb]
  \centering
  \subfigure[] {\includegraphics[width=0.28\textwidth]{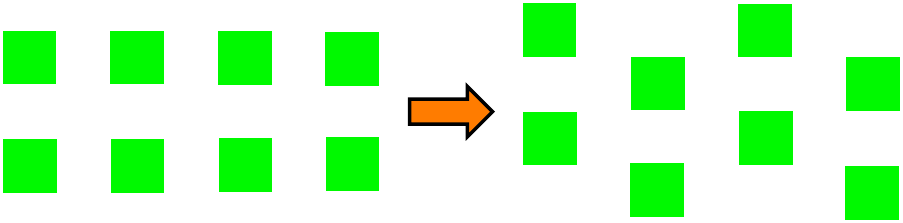}}
  \subfigure[] {\includegraphics[width=0.28\textwidth]{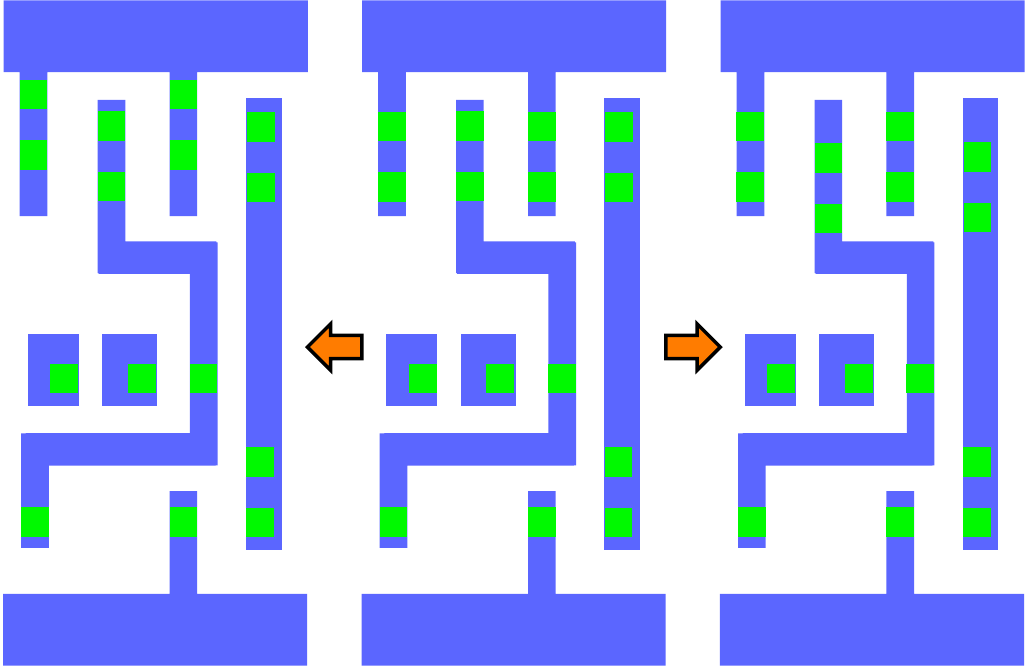}}
  \caption{~Contact layout modification to hexagonal packing. (a) The principle for contact shifting; (b) Demonstration of two options for contact shifting, with original layout in the middle, case 1 on the left and case 2 on the right.}
  %\vspace{-.1in}
  \label{fig:via_hexagonal}
\end{figure}

%\vspace{-.1in}
\subsection{Native TPL Conflict Removal}

%It is observed that some layout in given library may include common native conflicts for triple patterning, e.g., 4-clique conflict.
An example of native TPL conflict is illustrated in Fig. \ref{fig:via_hexagonal},
where four contacts introduce an indecomposable 4-clique conflict structure.
%The assumption is that there will be coloring conflict between the diagonal contacts in rectangular packing.
For such cases we modify the contact layout into hexagonal close packing \cite{TPL_SPIE2012_Lucas},
which also allows for the most aggressive cell area shrinkage for TPL friendly layout.
Note that after modification, the layout still needs to satisfy the design rules.
From the layout analysis of different cells, we have various ways to remove such 4-clique conflict. 
As shown in Fig. \ref{fig:via_hexagonal}, with slight modification to original layout, we can either choose to move contacts connected with power or ground rails or shift contacts on the signal paths of the cell. 
We call these two options case 1 and case 2 respectively, both of which will lead to TPL friendly standard cell layout.
%Note that usually the TPL conflict 

Generally, the cell layout design flexibility is beneficial for resolving conflicts between cells when they are placed next to each other. 
However, from a circuit designer's perspective, we want to achieve little timing variation among various layout styles of a single cell. 
Therefore, we need simulation results to demonstrate negligible timing impact from layout modification. 

\begin{figure}
  \centering
  \includegraphics[width=0.4\textwidth]{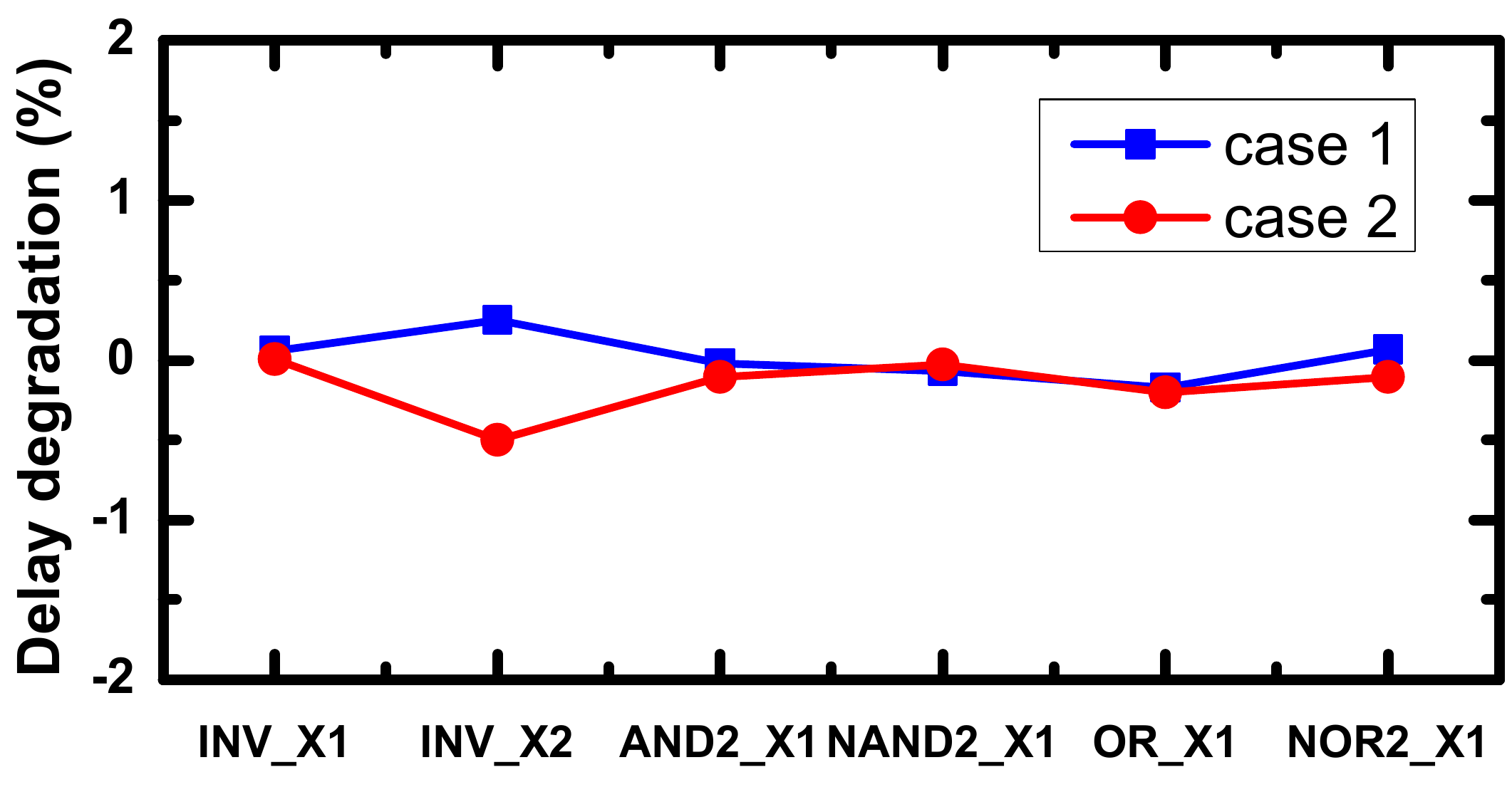}
  \caption{~The timing impact from layout modification for different types of gates, including case 1 and case 2}
  \label{fig:std-cell}
  %\vspace{-.1in}
\end{figure}

%\vspace{-.1in}
\subsection{Timing Characterization}

A Nangate 45nm Open Cell Library \cite{nangate} has been scaled to 16nm technology node.
After native TPL conflict detection and layout modification, we carry out the standard cell level timing analysis.
%Our standard cell analysis is based on Nangate 45nm Open Cell Library \cite{nangate}, which has been scaled it to 16nm technology node.
%The layout of standard cells has been scaled to 16nm technology node.
Calibre xRC \cite{Calibre} is used to extract parasitic information of the cell layout.
For each cell, we have original and modified layout with case 1 and case 2 options.
From extraction results, we can see that the source/drain parasitic resistance of transistors varies with the position of contacts, which is the direct impact from layout modification.
We use SPICE simulation to characterize different types of gates, which is based on 16nm PTM model \cite{ptm_16nm}.
Then, we can get the propagation delay of each gate, which is the average of rising and falling delay. 
%As mentioned previously, we care about the timing impact from layout modification.
We pick up six most commonly used cells to measure the relative changes of propagation delay due to layout modification
(see Fig. \ref{fig:std-cell}).
%Therefore, relative changes of propagation delay of each gate due to layout modification are illustrated in Fig. \ref{fig:std-cell}.
It is clearly observed that, for both case 1 and case 2, the timing impact will be within 0.5\% of the original propagation delay of gates,
which is assumed to be insignificant timing variation.
Based on case 1 or case 2 option, we will remove all conflicts among cells of the library with negligible timing impact.
Then we can ensure the standard cell compliance for triple patterning lithography.

%% file: doc/coloring.tex
\section{Standard Cell Pre-Coloring}
\label{sec:coloring}

For each type of standard cell, after removing the native TPL conflicts, we provide one set of pre-coloring solutions,
which can be prepared as a supplement to the library.
%Then in the next placement stage, one cell would be assigned one of it pre-coloring solutions as final color.
%This strategy can effectively resolve the conventional redundancy problem in layout decomposition, since for each cell there is only one solution search.
In this section we introduce the pre-coloring problem,
and propose general algorithms to solve it.
%{\color{blue} and propose general techniques to solve it}

%\vspace{-.1in}
\subsection{Problem Formulation}

At first glance, standard cell pre-coloring is similar to cell level layout decomposition.
However, different from the traditional layout decomposition, pre-coloring could have more than one solution for each cell.
It is observed that for some complex cell structure, if we exhaustively enumerate the possible coloring, it would have thousands of solutions.
Large solution size would impact the performance of our whole flow (see analysis in Section \ref{sec:singlerow}).
To provide high quality pre-coloring solutions, meanwhile keep the solution size as small as possible, we define immune feature and redundant coloring solutions as follows.

\begin{mydefinition}[Immune feature]
In one standard cell, an inside feature that would not conflict with any outside feature is defined as an immune feature.
%\vspace{-.1in}
\end{mydefinition}

It is easy to see that for one feature, if its distances to both vertical boundaries are larger than $d_{min}$, its color would not conflict with any other cells.
Then this feature is an immune feature.

\begin{mydefinition}[Redundant coloring solutions]
If two coloring solutions are only different at the immune features, these two solutions are redundant to each other.
%\vspace{-.1in}
\end{mydefinition}

\begin{problem}[Standard Cell Pre-Coloring]
Given the input standard cell layout, and the maximum allowed stitch number $maxS$,
we seek to search all coloring results that with stitch number no more than $maxS$.
%Meanwhile, we will remove the redundant coloring solutions.
Meanwhile, all redundant coloring solutions should be removed.
%\vspace{-.1in}
\end{problem}

%In other words, only those features that could conflict with other cell would be considered.
%(2) Temporally removing some inside features could reduce the number of pre-coloring solutions.
%Take the AND2X1 for example, initially there are 114 possible coloring solutions, while after removing the middle features (see Fig. \ref{fig:AND2X1_DG}(c)), there are only 24 possible solutions.
For example, given an AND2X1 cell as shown in Fig. \ref{fig:AND2X1_DG}(a),
if $maxS$ is set as $1$, the pre-coloring problem would search eight solutions (4 solutions with 0 stitch and 4 solutions with 1 stitch, see Fig. \ref{fig:AND2X1_solutions}).
%while traditional layout decomposition can only report one solution.

\begin{figure}
  \vspace{-.1in}
  \centering
  \subfigure[] {\includegraphics[width=0.16\textwidth] {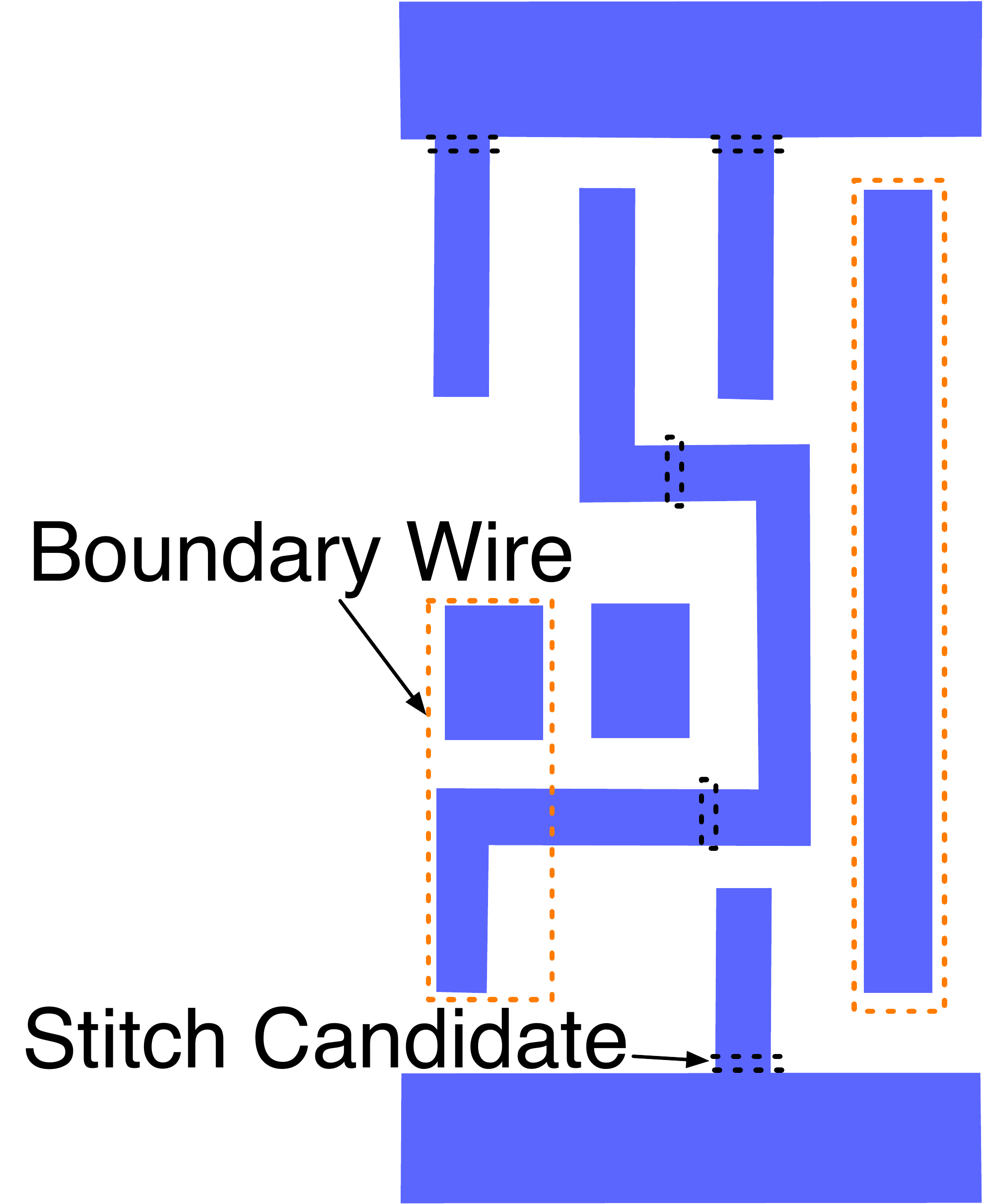}}
  \hspace{.1in}
  \subfigure[] {\includegraphics[width=0.095\textwidth]{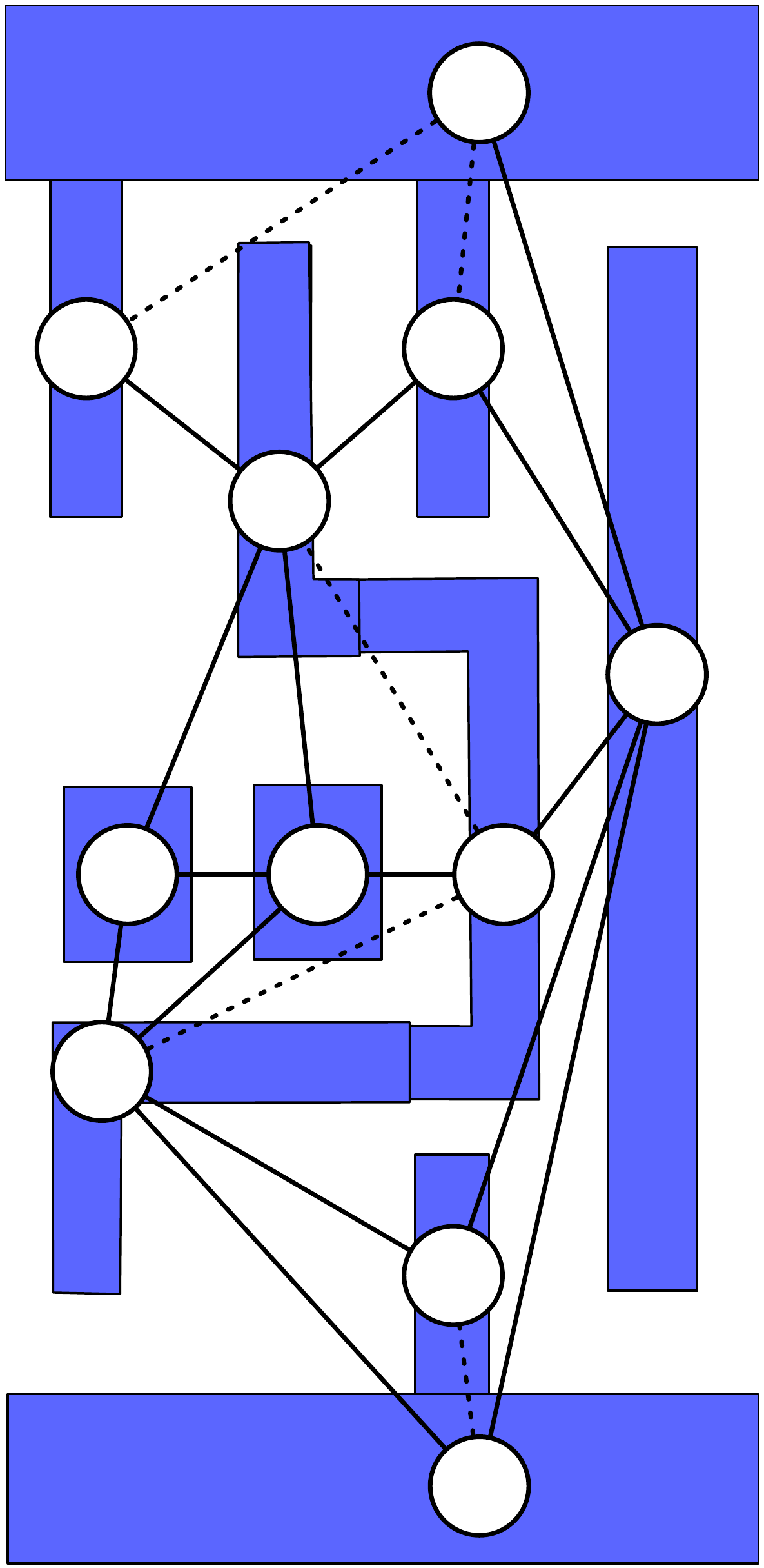}}
  \hspace{.1in}
  \subfigure[] {\includegraphics[width=0.095\textwidth]{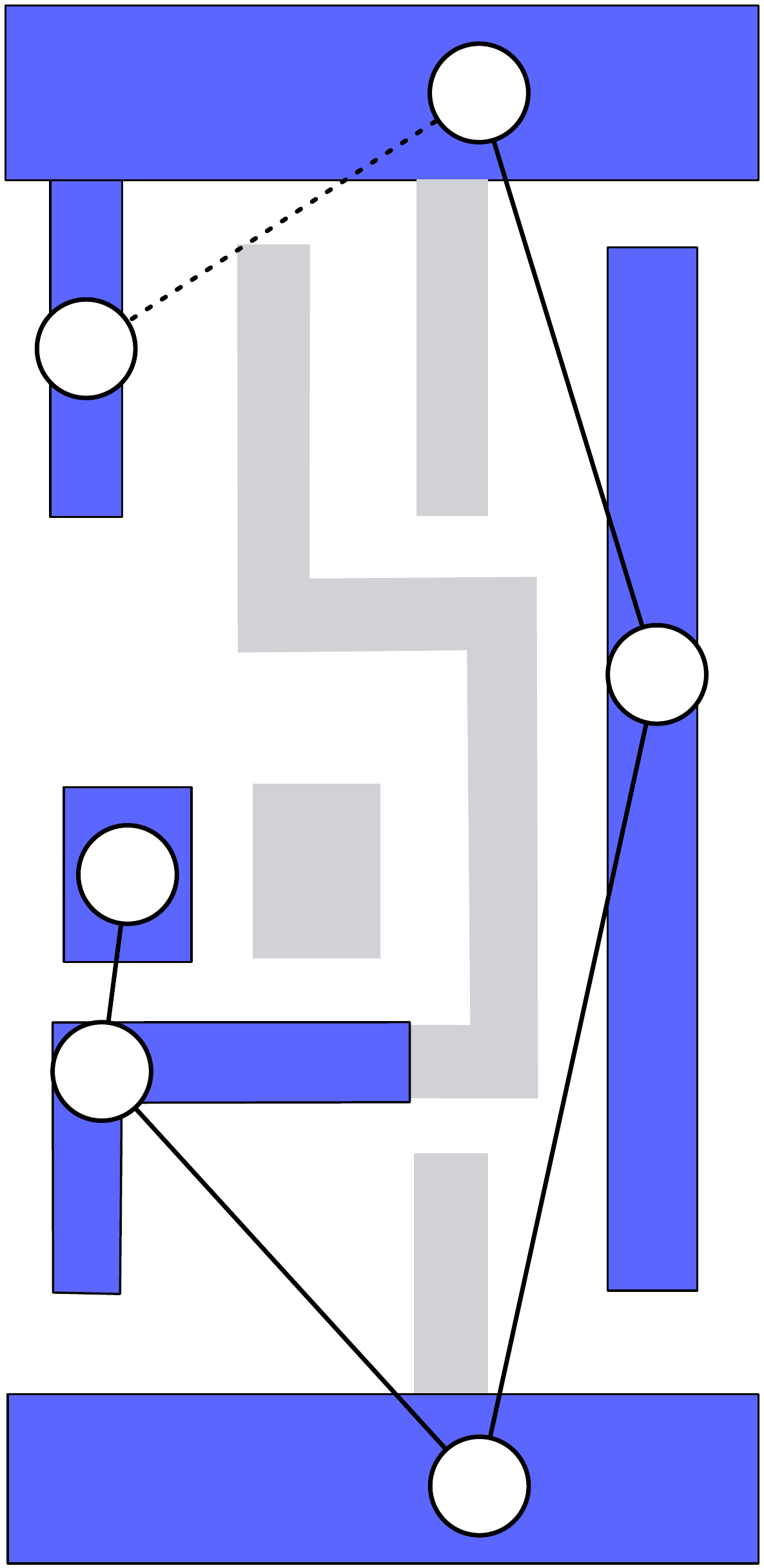}}
  \hspace{.2in}
  \caption{~Constraint graph construction and simplification.~(a) Input layout and all stitch candidates.
  ~(b) Constraint graph (CG) where solid edges are conflict edges and dash edges are stitch edges.
  ~(c) The simplified constraint graph (SCG) after removing immune features.}
  \label{fig:AND2X1_DG}
  %\vspace{-.1in}
\end{figure}

\begin{figure}[tb]
  \centering
  \hspace{-.1in}
  \subfigure[] {\includegraphics[width=0.24\textwidth]{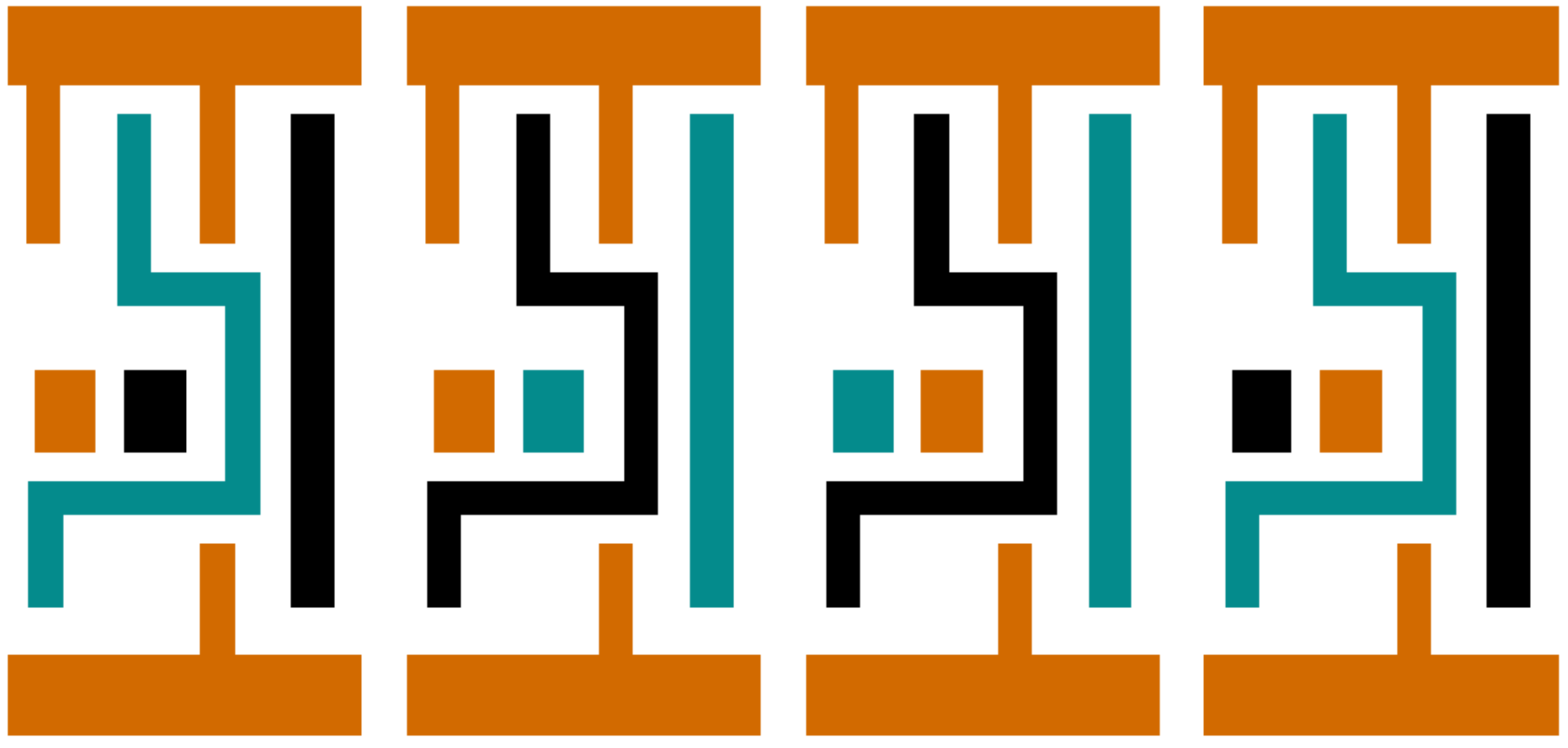}}
  \subfigure[] {\includegraphics[width=0.24\textwidth]{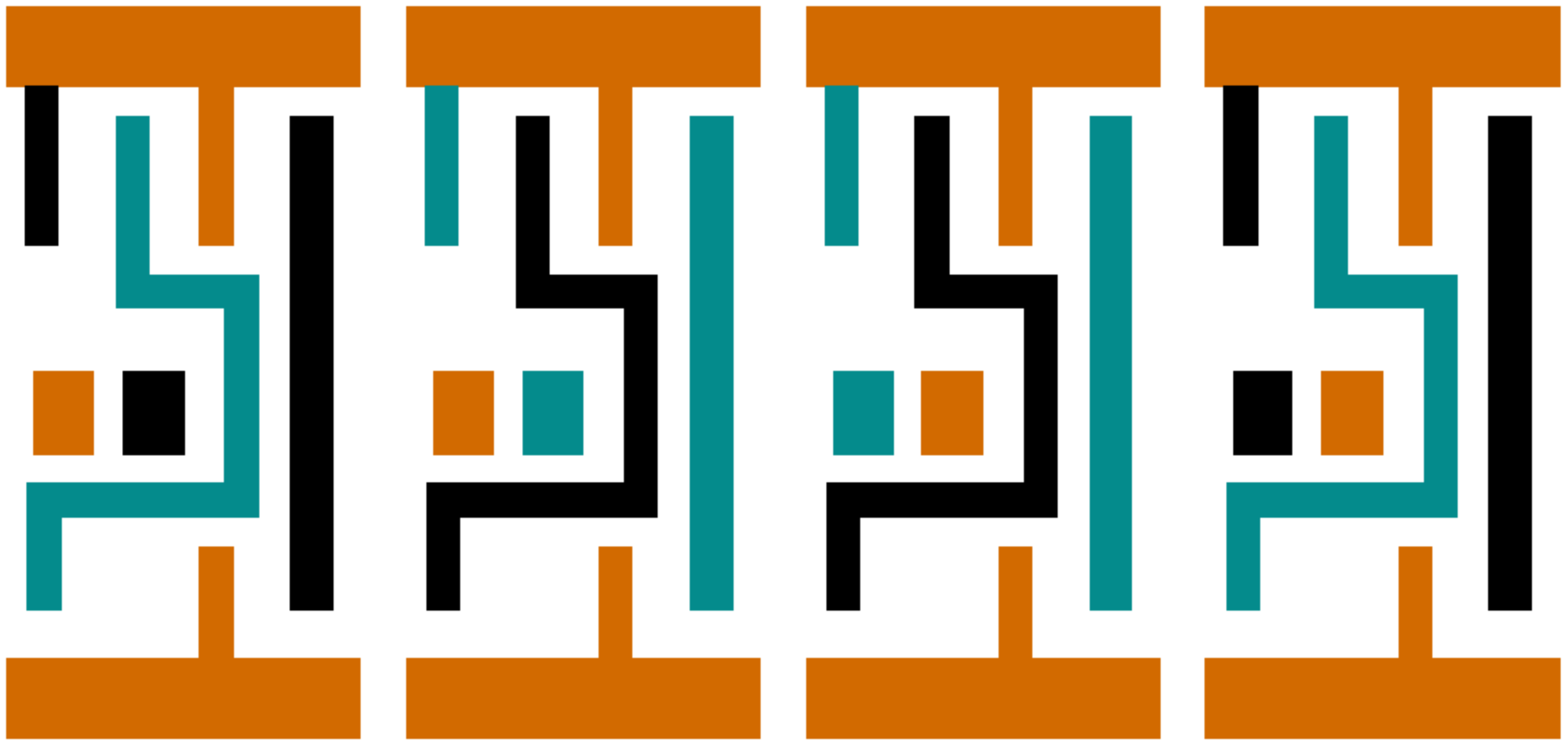}}
  \hspace{-.1in}
  \caption{~AND2X1 pre-coloring solutions with (a) 0 stitch; (b) 1 stitch.}
  \label{fig:AND2X1_solutions}
  \vspace{-.1in}
\end{figure}

Given the input standard cell layout, all the stitch candidates are captured through wire projection \cite{TPL_ICCAD2011_Yu}\cite{TPL_ICCAD2012_Tian}.
An example of AND2X1 cell is illustrated in Fig. \ref{fig:AND2X1_DG} (a), where five stitch candidates are captured for M1 layer.
Note that we forbid stitch on small features, e.g., contact, due to printability issue.
In addition, different from previous stitch candidate generation, we forbid the stitch on boundary metal wires (see the red boxes in Fig. \ref{fig:AND2X1_DG} (a)).
The reason is based on the observation that boundary stitches tend to cause indecomposable patterns between two cells.
Then an undirected constraint graph (\textit{CG}) \cite{TPL_ICCAD2011_Yu} is constructed to represent all input features and all the stitch candidates.
One feature in the layout is divided into two vertices in the graph if one stitch candidate is introduced.
The CG contains two sets of edges, i.e., the conflict edges and the stitch edges, respectively.
Fig. \ref{fig:AND2X1_DG} (b) shows the corresponding CG.

%Meanwhile, we can add stiches to increase the coloring solution space of the cell.
%If we look at the cell layout, we can find that some metal track is connected with the gate of transistor.
%We will avoid stich on that metal track because it will lead to potential timing penalty.
%Actually, we will minimize the number of stiches within cell and the only stitch we will introduce is on the power or ground tracks.
%The reason is that it helps to resolve conflicts between cells and we can see the benefit in the following section.

%\vspace{-.1in}
\subsection{SCG Solution Enumeration}

Since in CG some vertices represent the immune features, to avoid redundant coloring solutions, these features are temporarily removed.
We denote the remained graph as \textit{simplified constraint graph} (SCG).
%Since all these coloring solutions are required in our whole flow, reduce the redundancy can effectively speed-up the flow.
%If we take all these solutions into consideration, the detailed placement would induce timing penalty.
A backtracking algorithm \cite{book2002combinatorial} is proposed to the simplified CG to enumerate all possible coloring solutions.
For example, given the SCG shown in Fig. \ref{fig:AND2X1_DG}(c), there are 24 solutions.
%In a standard cell, the metal one layer includes power/ground rails and metal rails.
It should be mentioned that since all power/ground rails are assigned default color, the colors of corresponding vertices are assigned before the backtracking process.
%This means, during placement, the power or ground tracks will be the same color, which is a benefit to the power delivery of the design. 
%because only boundary contacts will induce additional conflicts when cells are placed next to each other.

% **********************************************************
%           Solution Verification
% **********************************************************
%\vspace{-.1in}
\subsection{CG Solution Verification}

Until now we have enumerated all coloring solutions for simplified constraint graph (SCG).
However, under the maximum stitch number $maxS$ constraint,
not all the SCG solutions can achieve legal layout decomposition in initial constraint graph (CG).
Therefore, CG solution verification is proposed to each generated solution.
Since SCG is a sub-set of CG, the verification can be viewed as layout decomposition with pre-colored features on SCG.
If a coloring solution for whole CG can be found with stitch number less than $maxS$, it would be stored as one pre-coloring solution.
%Given the coloring solutions, some vertices in initial constraint graph have been pre-colored.
%We need to search a color assignment for other vertices that the whole layout decomposition can achieve zero conflict and minimal stitch number.

% ****************************************
%  Algorithm: Solution Verification
% ****************************************
%{{{
\begin{algorithm}[htb]
\caption{CG Solution Verification}
\label{alg:verification}
\begin{algorithmic}[1]
  \Require set of initial coloring solutions $S'$ for SCG;
  \State Generate corresponding coloring solutions $S$ for CG;
  \For{ each coloring solution $s_i \in S$}
    \State $minCost \leftarrow \infty$;
    \State BRANCH-AND-BOUND($0, s_i$);
    \If {$minCost < maxS$}
      \State Output $s_i$ as legal pre-coloring solution;
    \EndIf
  \EndFor

  \Statex
  \Function{BRANCH-AND-BOUND}{$t, s_i$}
    \If {t $\ge \textrm{size}[s_i]$}
      \If { GET-COST( ) $< minCost$ }
        \State $minCost \leftarrow$ GET-COST();
      \EndIf
    \ElsIf {LOWER-BOUND( ) $> minCost$}
      \State Return;
    \ElsIf {$s_i[t] \ne -1$} %\Comment{vertex has been assigned color in Simplified DG}
      \State BRANCH-AND-BOUND($t+1, s_i$);
    \Else \Comment{$s_i[t] = -1$}
      \For {each available color $c$};
        \State $s_i[t] \leftarrow c$;
        \State BRANCH-AND-BOUND($t+1, s_i$);
        \State $s_i[t] \leftarrow -1$;
      \EndFor
    \EndIf
  \EndFunction
\end{algorithmic}
\end{algorithm}
%}}}

As shown in Algorithm \ref{alg:verification}, the CG solution verification is based on branch and bound \cite{book2002combinatorial}.
Given the coloring solutions $S' = \{s_1', s_2' \dots s_n'\}$ for SCG,
at the beginning the corresponding coloring solutions $S = \{s_1, s_2, \dots, s_n\}$ for CG are generated (line $1$).
Then we iteratively check each coloring solution $s_i$ (lines $2-6$).
For one coloring solution $s_i$, if vertex $t$ belongs to SCG, $s_i[t]$ should be already assigned one legal color.
If $t$ does not belong to SCG, $s_i[t] \leftarrow -1$.
The BRANCH-AND-BOUND() algorithm traverses the decision tree with a depth first search (DFS) methods (lines $7-19$).
For each vertex $t$, if $s_i[t]$ has been assigned one legal color in SCG, we skip $t$ and travel to the next vertex.
Otherwise, every legal color would be assigned to $t$ before traveling to the next vertex.
Different from exhaustive search, search space can be effectively reduced through pruning process (lines $11-12$).
%In other words, for a branch and bound based algorithm, tight lower bound is important for efficiency.
The function LOWER-BOUND() is to get lower bound by calculating current stitch number.
Note that if one conflict is found, then the function returns a large value.
Before checking any legal color of vertex $t$, we calculate its lower bound first.
If LOWER-BOUND() is larger than $minCost$, we shall not branch from $t$, since all the children solutions will be of higher cost than $minCost$.
Through the travel, all vertices have been assigned legal colors, stored in $s_i$.
After the travel, if $minCost \le maxS$, then $s_i$ is one of the pre-coloring solutions (lines $5-6$).

It shall be noted that although other optimal layout decomposition techniques, like integer linear programming (ILP), may be modified as the verification engine, our branch and bound based method is easy to implement and effective for standard cell level problem size.
Even for the most complex cell, SCG solution enumeration and CG solution verification can be finished in 5 seconds.

%\vspace{-.1in}
\subsection{Look-Up Table Construction}
\label{sec:lut}

%The library contains fixed number of standard cells.
For each cell $c_i$ in the library, we have generated a set of pre-coloring solutions $S_i = \{s_{i1}, s_{i2}, \dots, s_{iv}\}$.
We further pre-compute the decomposability of each cell pair, and store them in a lookup table.
For example, if two cells $c_i, c_j$ are assigned with $p-$th and $q-$th coloring solutions, respectively,
then LUT$(i, p, j, q)$ would store the minimum distance required when $c_i$ is to the left of $c_j$.
That is, if two colored cells can be legally abutted to each other, the corresponding value would be $0$.
Otherwise, the value would be the distance required to keep two cells decomposable.
Meanwhile, for each cell, its stitch number in different coloring solutions are also stored.
%Meanwhile, s$(i, p)$ stores the stitch number required when cell $c_i$ is assigned the $p-$th coloring solution.
It shall be noted that during the Lookup table construction, the cell flipping is considered, and related values are stored as well.

%% file: doc/singlerow.tex
\section{TPL aware Single Row Placement}
\label{sec:singlerow}

In this section we solve a single row placement, where the orders of all cells on the row are determined.
When TPL process is not considered, this row based design problem is called \textit{Ordered Single Row} (OSR) problem,
which has been well studied \cite{PLACE_ASPDAC99_Kahng, PLACE_DATE00_Brenner, PLACE_ICCAD05_Kahng}.
%DP \cite{PLACE_ASPDAC99_Kahng,PLACE_GLSVLSI04_Kahng,PLACE_ICCAD05_Kahng},
%LP and dual to network flow algorithm \cite{PLACE_DATE98_Vygen},
%and clumping algorithm \cite{PLACE_ASPDAC99_Kahng,PLACE_DATE00_Brenner}.
Here we revisit the OSR problem with the TPL process consideration.

% ************************************************
%             Problem Formulation
% ************************************************
%\vspace{-.1in}
\subsection{Problem Formulation}
%{{{
To formalize the OSR problem under TPL process, we introduce the following notations.
We consider an input single row as $m$ ordered sites $R = \{r_1, r_2, \dots, r_m\}$, and an input $n$ movable cells $C = \{c_1, c_2, \dots, c_n\}$ whose order is determined.
That is, $c_i$ is to the left of $c_j$, if $i < j$.
Each cell $c_i$ has $v_{i}$ different coloring solutions.
A \textit{cell-color pair} $(i, p)$ denotes that cell $c_i$ is assigned to the $p-$th color solution, where $p \in [1, v_i]$.
Meanwhile, $s(i, p)$ gives the corresponding stitch number for $(i, p)$.
The horizontal position of cell $c_i$ is given by $x(c_i)$, and the cell width is given by $w(c_i)$.
All the cells in other rows are with fixed positions.
A single row placement is \textit{legal} if and only if any two cells $c_i, c_j$ meets the following non-overlap constraint:
%if $c_i$ and $c_j$ are assigned color solutions $s_{ip}$ and $c_{jq}$ respectively, then
\begin{displaymath}
  x(c_i) + w(c_i) + \textrm{LUT}(i, p, j, q) \le x(c_j), \ \ \textrm{if } (i, p) \& (j, q)
\end{displaymath}
where LUT($i, p, j, q$) is corresponding LUT value mentioned in Section \ref{sec:lut}.
Based on all these notations, we define the TPL aware Ordered Single Row (\textit{TPL-OSR}) problem as follows.

\begin{figure}[tb]
  \centering
  \subfigure[] {\includegraphics[width=0.36\textwidth]{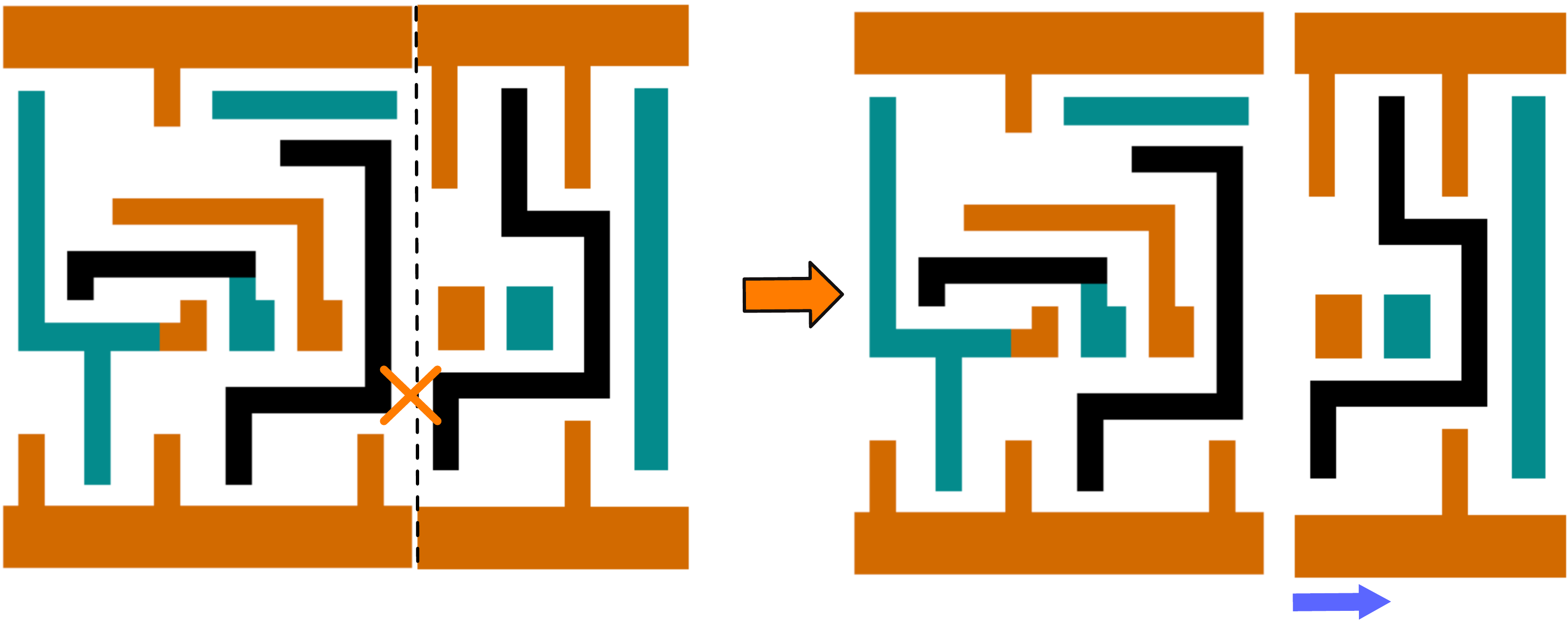}}
  \subfigure[] {\includegraphics[width=0.36\textwidth]{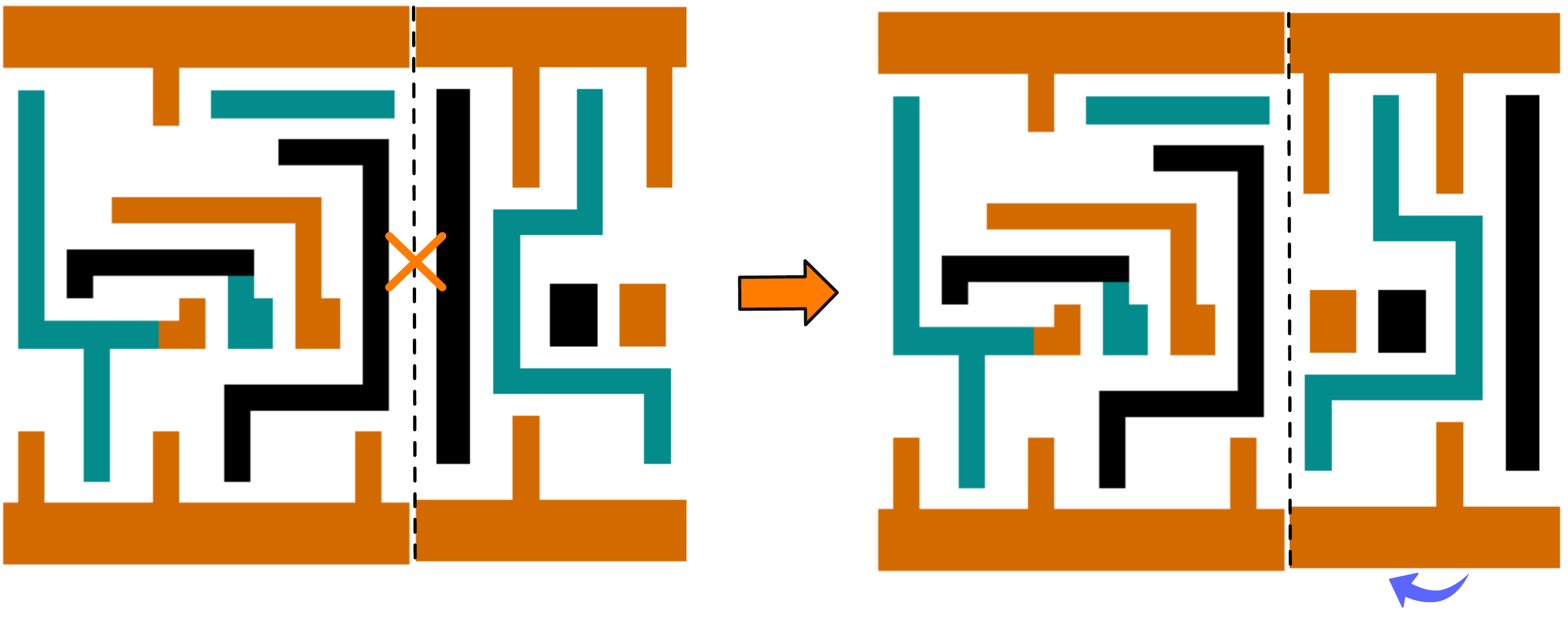}}
  \caption{~Two techniques for removing conflicts during placement. (a) Flip the cell; (b) Shift the cell.
  %Note: the cell layout is used for demonstration only
  }
  \label{fig:between_cells_metal}
  %\vspace{-.1in}
\end{figure}

\begin{problem}[TPL aware Ordered Single Row Problem]
Given a single row placement, we seek a legal placement and cell color assignment, so that the half-perimeter wire-length (HPWL) of all nets and the total stitch number are minimized.
\end{problem}

Compared with traditional OSR problem, TPL-OSR problem faces two special challenges:
(1) TPL-OSR not only needs to solve cell placement, but also needs to assign appropriate coloring solutions for cells to minimize the stitch number.
In other words, cell placement and color assignment should be solved simultaneously.
(2) In conventional OSR problem, if the sum of all cell width is less than row capacity, it is guaranteed that there would be one legal placement solution.
However, for TPL-OSR problem, since some extra sites may be spared to resolve coloring conflicts, before coloring assignment we cannot calculate the required site number.

In addition, it shall be noted that compared with conventional color assignment problem, in TPL-OSR the solution space is much larger.
That is, to resolve the coloring conflict between two abutted cells $c_i, c_j$, apart from picking up compatible coloring solutions,
TPL-OSR can seek to flip cells (see Fig. \ref{fig:between_cells_metal} (a)) or shift cells (see Fig. \ref{fig:between_cells_metal} (b)).

\begin{comment}
\subsection{ILP Formulation}
%{{{
We denote critical nets as $N = \{n_1, n_2, \dots, n_p\}$.
We consider each row as set of ordered sites $S = \{s_1, s_2, \dots, s_n\}$.
For each row $r_i$, our goal is to place its set of cells $C_i$ with no overlap.

\begin{align}
  \label{eq:ilp}
  \textrm{min}\ \ \  &  \sum_{n_p \in N} (R_p - L_p)   \\
  \textrm{s.t}.\ \ \ &  x_i + w_i \le R      & \forall i \in C \notag\\
                     &  x_i \ge L            & \forall i \in C \notag\\
                     &  x_i + w_i \le R_p    & \forall i \in n_p \notag\\
                     &  x_i \ge L_p          & \forall i \in n_p \notag\\
                     &  x_i + w_i + d_{ij} \le x_j & \forall \{ij\} \notag\\
                     &  d_{ij} = \sum_{k} \sum_{l} d_{ijkl} \cdot c_{ik} \cdot c_{jl} \notag\\
                     &  \sum_{k} c_{ik} = 1  & \forall i \in C \notag
\end{align}

The ILP formulation can simultaneously assign colors to all cells, and determine the positions of them.
Since ILP is a well-known NP-hard problem, it suffers from long runtime penalty to achieve the final results.
%}}}
\end{comment}
%}}}

% ************************************************
%            Shortest-Path Algorithm
% ************************************************
%\vspace{-.1in}
\subsection{Graph Model for TPL-OSR}

In this subsection we propose a graph model that correctly captures the cost of HPWL and the stitch number.
Furthermore, we will show that performing a shortest path algorithm on the graph model can optimally solve the TPL-OSR problem.

To consider cell placement and cell color assignment simultaneously, a directed acyclic graph $G = (V, E)$ is constructed.
The graph $G$ is with vertex set $V$ and edge set $E$.
$V = \{\{0, \dots, m\} \times \{0, \dots, N\}, t\}$, where $N = \sum_{i=1}^{n} v_i$.
The vertex in the first row and the first column is defined as vertex $s$. 
We can see that each column corresponds to one site's start point, and each row is related to one specified color assignment of one cell.
Without loss of generality, we label each row as $r(i, p)$ if it is related to cell $c_i$ with $p-$th coloring solution.
%For each vertex $n_i$ we denote its cell and color assignment as $c(n_i)$ and $a(n_i)$, respectively.
%We label each vertex as pair $(c_i, s_p)$. 
The edge set $E$ is composed of three sets of edges: horizontal edges $E_h$, ending edges $E_e$, and diagonal edges $E_d$.
\begin{align}
E_h =  & \{(i, j-1) \rightarrow (i, j) | 0 \le i \le N, 1 \le j \le m\}   \notag\\
E_e =  & \{(i, m) \rightarrow t | i \in [1, N]\}                              \notag\\
E_d =  & \{(r(i-1,p), k) \rightarrow (r(i,q), k+w(c_i)+                 \notag\\
            & LUT(i-1, p, i, q)) | i \in [2, n], p \in [1, v_{i-1}], q \in [1, v_i] \} \notag
\end{align}
We denote each edge by its start and end point.
A legal TPL-OSR solution corresponds to finding a directed path from the vertex $s$ to vertex $t$.
Sometimes one row cannot insert all the cells, therefore ending edges are introduced.
With these ending edges, the graph model can guarantee to find out one path from $s$ to $t$.

\begin{figure}[tb]
  \centering
  \includegraphics[width=0.44\textwidth]{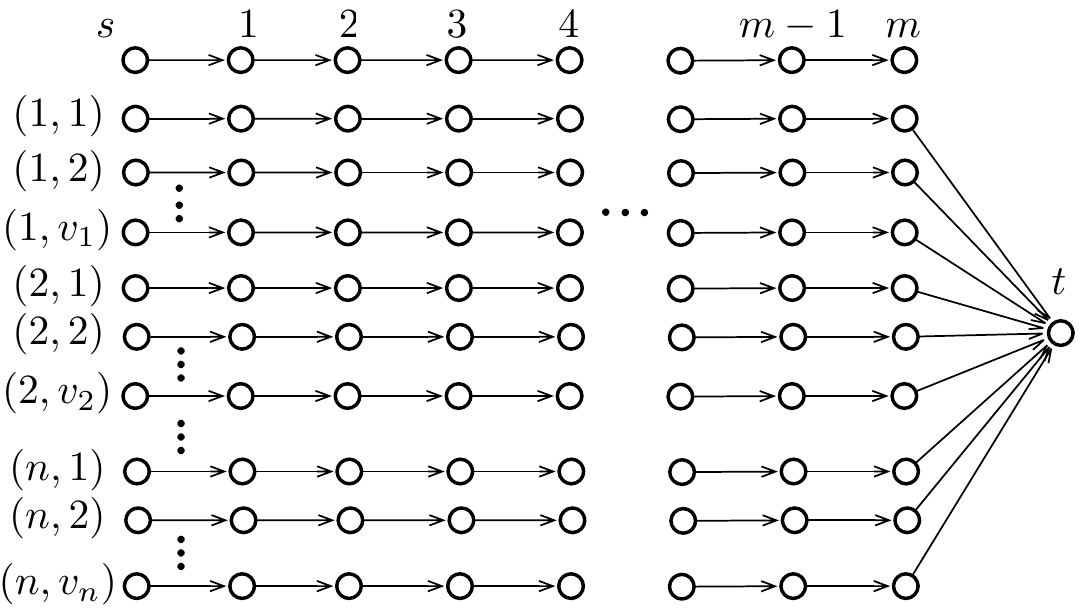}
  \caption{~Graph model for the TPL-OSR problem (only the horizontal edges and ending edges are showed).}
  \label{fig:graph_singlerow}
  %\vspace{-.1in}
\end{figure}

\begin{figure*}[tb]
  \centering
  \hspace{-.1in}
  \subfigure[] {\includegraphics[width=0.20\textwidth]{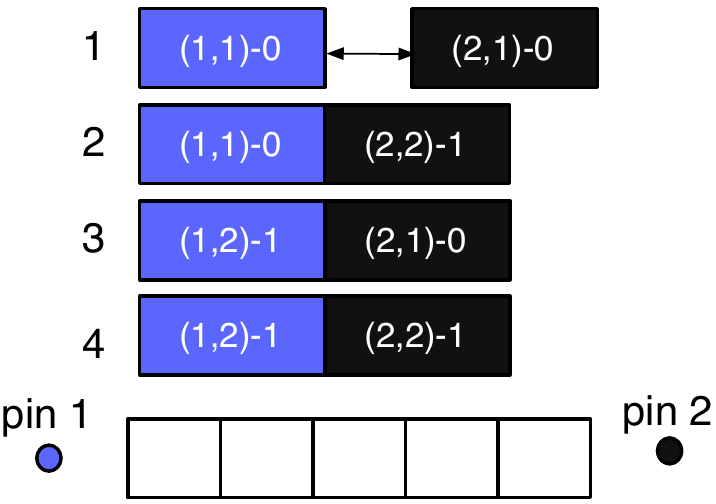}}
  \subfigure[] {\includegraphics[width=0.26\textwidth]{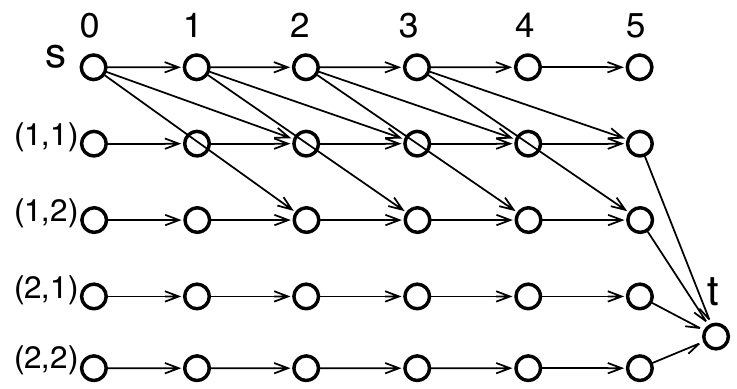}}
  \subfigure[] {\includegraphics[width=0.26\textwidth]{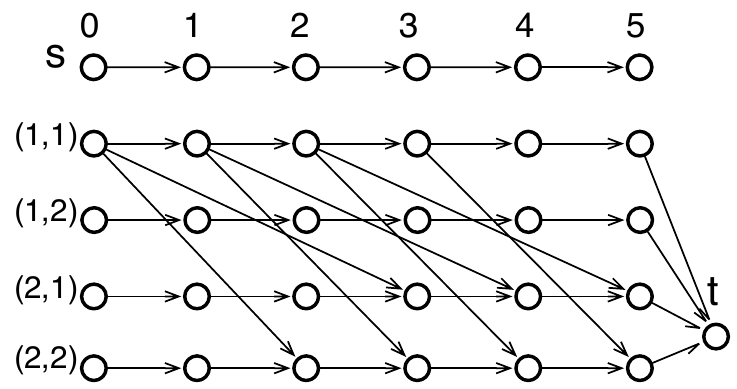}}
  \subfigure[] {\includegraphics[width=0.26\textwidth]{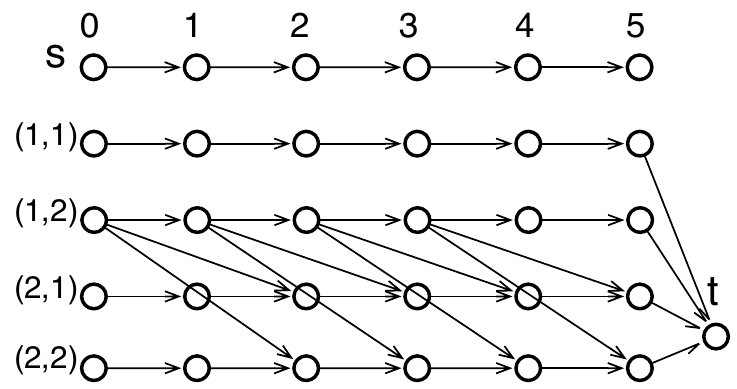}}
  \hspace{-.2in}
  \caption{~Example for the TPL-OSR problem.~(a) two cells with different coloring solutions to be placed into a 5 sites row;
                Graph models with diagonal edges (b) from s vertex to first cell; (c) from c1\_1 to second cell; (d) from c1\_2 to second cell.}
  \label{fig:example_graph}
\end{figure*}

To simultaneously minimize the HPWL and stitch number, we define the cost on edges as follows.
(1) All horizontal edges are with zero cost.
(2) For ending edge $\{(r(i, p), m) \rightarrow t\}$, it is labelled by the cost $(n-i) \cdot M$, where $M$ is a large number.
(3) For diagonal edge $\{(r(i, p), k) \rightarrow (r(j,q), k+w(c_j) + LUT(i, p, j, q))\}$, it is labelled by the cost as follows:
\begin{displaymath}
\Delta WL + \alpha \cdot s(i, p) + \alpha \cdot s(j, q)
\end{displaymath}
where $\Delta WL$ is the HPWL increment of placing $c_j$ in position $q - \textrm{LUT}(i, p, j, q)$.
Here $\alpha$ is a user-defined parameter for assigning relative importance between the HPWL and the stitch number.
In our framework, $\alpha$ is set as 10.
The general structure of $G$ is shown in Fig. \ref{fig:graph_singlerow}.
Note that for clarity here we do not show the diagonal edges.
%While any such path produces a legal placement, it is highly desirable to minimize the impact on placement metrics.
%We define the objectives on edges as follows.

\begin{figure}[tb]
  \centering
  \hspace{-.2in}
  \subfigure[] {\includegraphics[width=0.245\textwidth]{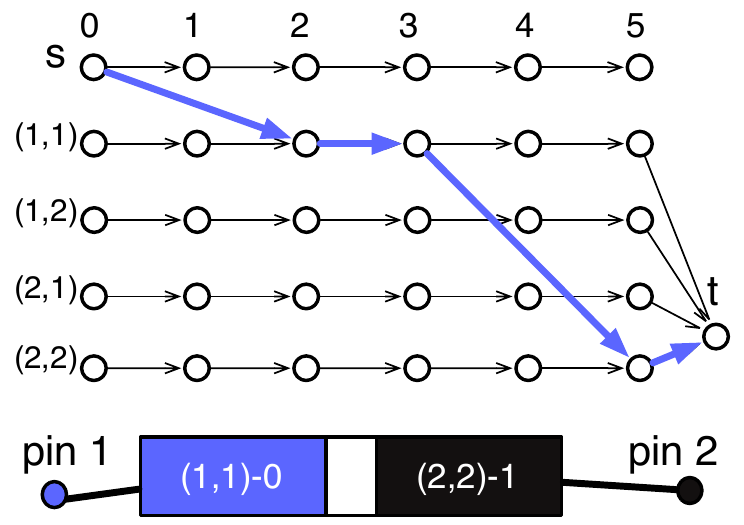}}
  \subfigure[] {\includegraphics[width=0.245\textwidth]{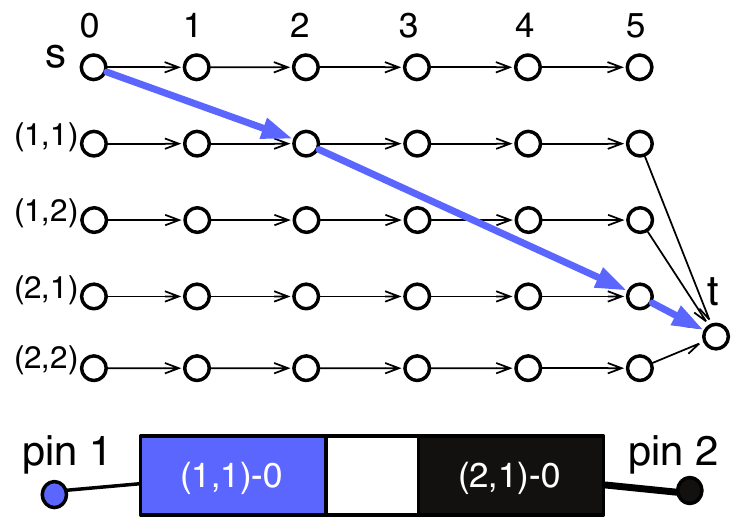}}
  \hspace{-.2in}
  \caption{Shortest path solutions on the graph model with (a) 1 stitch; (b) 0 stitch.}
  \label{fig:example_path}
\end{figure}

One example of the graph model is illustrated in Fig. \ref{fig:example_graph},
where two cells $c_1$ and $c_2$ are to placed in a row with 5 sites.
Each cell has two different coloring solutions, and corresponding required stitch number.
For example, the label (2,1)-0 means $c_2$ is assigned to the first coloring solution, with no stitch.
The graph model is shown in Fig. \ref{fig:example_graph}(b)(c)(d), where each figure shows different part of diagonal edges.
Cells $c_1$ and $c_2$ are connected with pin 1 and pin 2, respectively.
Therefore, $c_1$ tends to be on the left side of row, while $c_2$ tends to be on the right side.
Fig. \ref{fig:example_path} gives two shortest path solutions with the same HPWL.
Because the second one is with less stitch number, it would be selected as the solution for TPL-OSR problem.
%optimacell placement and coloring assignment.

Since $G$ is a directed acyclic graph, the shortest path can be calculated using topological traversal of $G$ in $O(mnK)$ steps,
where $K$ is the maximal pre-coloring solution number for each cell.
To apply topological traversal, a dynamic programming algorithm is proposed to find the shortest path from the $s$ vertex to the $t$ vertex.

%\vspace{-.1in}
\subsection{Two Stage Speedup}

\begin{figure}[htb]
  \centering
  \subfigure[]{\includegraphics[width=0.23\textwidth]{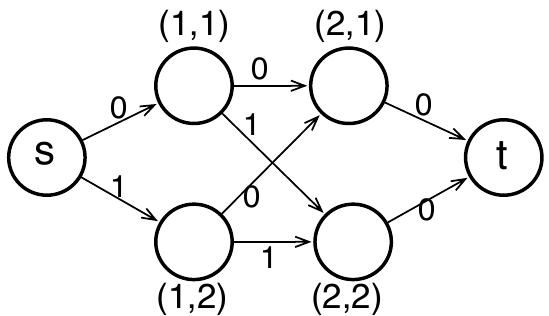}}
  \subfigure[]{\includegraphics[width=0.23\textwidth]{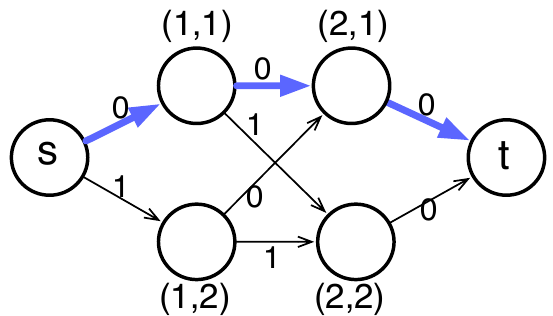}}
  \caption{First stage model to solve color assignment. In this example edge cost only considers the stitch number minimization.}
  \label{fig:model_stage1}
  %\vspace{-.1in}
\end{figure}

Although the shortest path algorithm can be solved in $O(mnK)$, for practical design when each cell could allow many pre-coloring solutions, the proposed graph model may still suffer from long runtime penalty.
To achieve a better trade-off between runtime and performance, here we propose a two-stage speedup technique.
%The main idea is that the whole shortest path algorithm is divided into two sub-problems, i.e., color assignment and ordered single placement.
%Correspondingly the previous graph model is decomposed into two graph models, one for color assignment, another for cell placement.
The main idea is that the whole previous graph model is decomposed into two smaller graph models, one for color assignment, and another for cell placement.

To solve the example in Fig. \ref{fig:example_graph}, the first stage graph model is illustrated in Fig. \ref{fig:model_stage1} (a), where the cost of each edge corresponds to the stitch number required for each cell-color pair $(i, p)$.
Note that in our framework, relative positions among cells are also considered in the edge cost.
A shortest path on the graph corresponds to a color assignment with minimum stitch number.

Our second stage is for cell placement, and the previous color assignment solutions are considered here.
That is, if in previous color assignment cells $c_{i-1}$ and $c_i$ are assigned its $p-$th and $q-$th coloring solutions,
then the width of cell $c_{i}$ is changed from $w(c_i)$ to $w(c_i) + \textrm{LUT}(i-1, p, i, q)$.
By this way, the extra site to resolve coloring conflicts are prepared for cell placement.
Based on the updated cell widths, the graph model in \cite{PLACE_ICCAD05_Kahng} can be directly applied here.

The first graph model can be solved in $O(nK)$, while the second graph model can be resolved in $O(mn)$.
Therefore, although the speedup technique can not achieve optimal solution of TPL-OSR problem,
applying the two-stage graph model can reduce the complexity from $O(mnK)$ to $O(nK + mn)$.

%% file: doc/dplace.tex
%\vspace{-.1in}
\section{Overall Placement Scheme}
\label{sec:dplace}

In this section we present our overall scheme for the TPL aware detailed placement.
%We adopt a negotiated congestion based scheme to resolve the coloring conflicts over iterations of rip-up and reroute/re-color.

\begin{algorithm} [tb]
\caption{TPL aware Detailed Placement}
\label{alg:dplace}
\begin{algorithmic}[1]
  \Require cells to be placed;
  \Repeat
  \State Sort all rows;
  \State Label all rows as $FREE$;
  \For {each row $row_i$}
    \State Solve TPL-OSR prolbem for $row_i$;
    \If {exist unsolved cells}
      \State Global Moving;
      \State Update cell widths considering assigned colors;
      \State Solve cell placement (OSR) for $row_i$;
    \EndIf
    \State Label $row_i$ as $BUSY$;
  \EndFor
  \Until{no significant improvement}
\end{algorithmic}
\end{algorithm}

The flow of our detailed placement algorithm is summarized in Algorithm \ref{alg:dplace}.
In each main loop, rows are sorted that the row with more cells occupied would be solved earlier.
At the beginning, all rows are labeled as $FREE$, which means it can be inserted additional cells (line $3$).
For each row $row_i$, we propose TPL-OSR algorithm as introduced in Section \ref{sec:singlerow} to solve color assignment and cell placement simultaneously.
Note that sometimes TPL-OSR cannot guarantee to assign all cells into row, due to extra sites required to resolve coloring conflicts.

If TPL-OSR ends with unsolved cells, \textit{Global Moving} is applied to move some cells to other rows (line $7$).
The basic idea behind the Global Moving is to find the ``optimal row and site'' for a cell in the placement region, and remove some local triple patterning conflicts.
For each cell we define its ``optimal region'' as the site to place where the HPWL is optimal \cite{PLACE_TCAS81_Goto}.
Note that one cell can be only moved to $FREE$ rows.
Since some cells in the middle of row may be moved, we need to solve OSR problem to rearrange the cell positions \cite{PLACE_ICCAD05_Kahng}.
Note that since all cells on the row have been assigned colors, cell widths should be updated to preserve extra space for coloring conflict (line $8-9$) .
After solving one $row_i$, it is labeled as $BUSY$ (line $10$).

Since the rows are placed and colored one by one sequentially, the solution obtained within one single row may not be good enough,
Therefore, our scheme is able to repeatedly call the main loop, until no significant improvement is achieved.

\begin{comment}
%\vspace{-.1in}
\subsection{Balancing Features on Three Masks}
TPL provides three masks to accommodate the features.
Balancing the features on the three masks ensures that each mask is fully utilized, so that none of the masks is unnecessarily dense.
We also propose heuristic method to balance the features over the masks, which is beneficial for manufacturing.
For one row, the colors for power/ground rails are assigned as default colors (color 1), but another two colors are free to rotate.
For example, if relatively more wires are assigned the color 3, after solve TPL-OSR for one row, the color 2 and color 3 can be rotated that more wires would be assigned to color 2.
\end{comment}

%% file: doc/result.tex
%\vspace{-.1in}
\section{Experimental Results}
\label{sec:result}

%  Experiment Setup
We implement our standard cell pre-coloring and TPL aware detailed placement in C++,
and all the experiments are performed on a Linux machine with 3.0GHz CPU.
We use Design Compiler \cite{design_compiler} to synthesize OpenSPARC T1 designs based on Nangate 45nm standard cell library \cite{nangate}.
During benchmark generation, all the library and benchmark are scaled down to $16 nm$ technology node.
%{\color{blue} For simplicity, we assume the sizes of the minimum pattern width, spacing, and spacer width are the same (?).}
%Cells are decomposed as explained and used to configure the decomposability look-up table.
For each benchmark, we perform placement with Cadence SOC Encounter \cite{socEncounter} to generate initial placement result.
To better compare the performance of detailed placement under different placement densities, 
for each circuit, we choose three different core utilization rates 0.7, 0.8, and 0.9.
%The benchmark information and our results are shown in Table 1.

% Comparisons
\begin{table*}[tb]
%{{{
\centering
\scriptsize
\renewcommand{\arraystretch}{1.1}
\caption{Comparisons of Detailed Placement Algorithms}
\label{tab:compare}
\begin{tabular}{|c|c|c|c||c|c|c|c|c|c|c|c|c|c|c|c|}
  \hline \hline
  bench   &\multicolumn{3}{c||}{Post-Decomposition}  &\multicolumn{4}{c|}{GREEDY \cite{DFM_SPIE2013_Gao}}
          &\multicolumn{4}{c|}{TPLPlacer}                               &\multicolumn{4}{c|}{TPLPlacer-SPD}\\
  \cline{2-16}  &CN\# &ST\#  &CPU(s)      &CN\# &ST\# &WLD  &CPU(s)        &CN\# &ST\# &WLD &CPU(s)  &CN\# &ST\# &WLD &CPU(s) \\
  \hline
  alu-70  &605  &4092  &0.7   &0   &1254  &+0.06\%  &2.0      &0 &1013 &-0.94\% &107.2        &0 &994  &-0.77\% &4.2    \\
  alu-80  &656  &4100  &0.6   &N/A &N/A   &N/A      &N/A      &0 &1011 &-1.70\% &114.8        &0 &994  &-1.48\% &4.6    \\
  alu-90  &596  &3585  &0.5   &N/A &N/A   &N/A      &N/A      &0 &1006 &-2.38\% &120.2        &0 &994  &-2.2\%  &4.8    \\
  byp-70  &1683 &9943  &2.4   &0   &3254  &0.14     &0.97     &0 &2743 &-5.98\% &382.5        &0 &2545 &-5.69\% &9.2    \\
  byp-80  &1918 &10316 &2.6   &N/A &N/A   &N/A      &N/A      &0 &2889 &-2.58\% &343.0        &0 &2545 &-2.12\% &7.9    \\
  byp-90  &2285 &10790 &3.0   &N/A &N/A   &N/A      &N/A      &0 &3136 &+1.74\% &361.9        &0 &2514 &+4.31\% &7.1    \\
  div-70  &1329 &6017  &2.2   &0   &2368  &+0.08\%  &1.89     &0 &2119 &-3.84\% &117.6        &0 &2017 &-3.28\% &5.6    \\
  div-80  &1365 &5965  &2.1   &0   &2379  &+0.08\%  &1.87     &0 &2090 &-2.06\% &135.6        &0 &2017 &-1.63\% &6.1    \\
  div-90  &1345 &5536  &2.1   &0   &2365  &+0.02\%  &1.87     &0 &2080 &-4.79\% &155.2        &0 &2017 &-4.37\% &6.4    \\
  ecc-70  &206  &3852  &0.9   &N/A &N/A   &N/A      &N/A      &0 &247  &-4.76\% &69.4         &0 &228  &-4.6\%  &1.7    \\
  ecc-80  &265  &3366  &1.0   &0   &433   &+0.43\%  &0.44     &0 &274  &-2.51\% &58.2         &0 &228  &-2.19\% &1.5    \\
  ecc-90  &370  &4015  &1.1   &N/A &N/A   &N/A      &N/A      &0 &369  &-1.28\% &68.5         &0 &228  &-1.53\% &1.4    \\
  efc-70  &503  &3333  &0.7   &0   &1131  &+0.0 \%  &5.46     &0 &1005 &-1.32\% &32.4         &0 &1005 &-1.31\% &6.2    \\
  efc-80  &570  &4361  &0.6   &N/A &N/A   &N/A      &N/A      &0 &1008 &-3.35\% &37.7         &0 &1005 &-3.26\% &6.3    \\
  efc-90  &534  &4040  &0.8   &0   &1133  &+0.0 \%  &5.4      &0 &1005 &+0.35\% &39.0         &0 &1005 &+0.35\% &6.3    \\
  ctl-70  &425  &2583  &1.3   &0   &703   &+0.23\%  &3.8      &0 &573  &-1.75\% &67.3         &0 &553  &-1.56\% &5.3    \\
  ctl-80  &529  &3332  &1.4   &0   &714   &+0.5 \%  &3.8      &0 &561  &-2.26\% &78.8         &0 &553  &-2.04\% &5.5    \\
  ctl-90  &519  &3241  &1.5   &0   &726   &+0.4 \%  &3.8      &0 &556  &-0.63\% &85.4         &0 &553  &-0.5\%  &5.6    \\
  top-70  &5893 &27981 &9     &N/A &N/A   &N/A      &N/A      &0 &8069 &-10.6\% &1948.0       &0 &8034 &-10.4\% &43.5   \\
  top-80  &6775 &32352 &10.3  &N/A &N/A   &N/A      &N/A      &0 &8120 &-5.45\% &1696.7       &0 &8015 &-4.9\%  &36.8   \\
  top-90  &7313 &29343 &11.4  &N/A &N/A   &N/A      &N/A      &0 &8710 &-4.41\% &1850.9       &0 &7876 &+2.09\% &32.7   \\
  \hline
  Average &\textbf{1700} &8664  &2.68  &N/A&N/A&N/A&N/A       &\textbf{0} &2314 &-2.88\% &142.9        &\textbf{0} &2186 &-2.24\% &9.94   \\
  Ratio   &&& &&&& &&\textbf{1.0}&\textbf{}&\textbf{1.0}      &&\textbf{0.95}&\textbf{}&\textbf{0.07} \\     
  \hline \hline
\end{tabular}
%}}}
\end{table*}

We compare four different design flows for M1 layer of all the benchmarks.
``\textbf{Post-Decomposition}'' means the conventional TPL design flow, where Encounter is chosen as the placer and an academic decomposer \cite{TPL_DAC2012_Fang} is applied for layout decomposition.
Note here the standard cell inner native conflicts have been removed through our compliance techniques (see Section \ref{sec:cell}).
We implement the greedy based detailed placement algorithm in \cite{DFM_SPIE2013_Gao}, denoted as ``\textbf{GREEDY}''.
Although the work in \cite{DFM_SPIE2013_Gao} is for the self-aligned double patterning (SADP) friendly design,
the proposed detailed placement algorithm can be integrated into our framework as well.
``\textbf{TPLPlacer}'' and ``\textbf{TPLPlacer-SPD}'' are the proposed detailed placement, where the ``TPLPlacer'' applies the optimal graph model to solve cell placement and color assignment simultaneously, while the ``TPLPlacer-SPD'' uses fast two-stages graph models to solve color assignment and cell placement iteratively.

All the experimental results are listed in Table \ref{tab:compare}, where columns ``CN\#'' and ``ST\#'' are the conflict number and the stitch number on the final decomposed layout, respectively.
Column ``WLD'' shows the total wire length difference between initial placement and our TPL aware placement, and column ``CPU(s)'' gives the runtime.
First, from the table we can see that under the conventional design flow, which is placement + layout decomposition,
even each standard cell itself is TPL-friendly, averagely 1,700 conflicts are reported in final decomposed layout.
%While our proposed flow, although no layout decomposition is applied, can directly generate decomposed layout without any conflict.
Second, we can see that although for some cases GREEDY can achieve 0 conflict results, in 10 out of 21 cases it cannot find out legal placement results.
For those illegal results labels ``N/A'' are reported.
The main reason is that GREEDY only shifts the cells toward right direction.
For some benchmark with high cell utilization, it may cause the final placement violation.
In addition, GREEDY uses a greedy method to assign cell color, thus it loses the global view to minimize the stitch number.
%Therefore, even for the cases it find out legal result, they report more stitch number.
Therefore, more stitches are reported for those cases where it finds out legal results.

We further compare our two detailed placement strategies, TPLPlacer, and TPLPlacer-SPD.
From Table \ref{tab:compare} we can see that TPLPlacer can achiever slightly better HPWL (0.22\%).
However, TPLPlacer-SPD can achieve 14x speedup and less stitch number (5\% less).
The speedup is due to the faster two-stage graph model, instead of combined graph model.
The main reason for the less stitch number is that the TPLPlacer-SPD solves color assignment first, followed by cell placement.
Although cell position is integrated into the color assignment model, the shortest path with less number of stitches tends to be selected.
The results in Table \ref{tab:compare} demonstrate the effectiveness of our standard cell compliance and detailed placement techniques.

%% file: doc/conclu.tex
%\vspace{-.1in}
\section{Conclusion}
\label{sec:conclu}

In this paper we propose a coherent framework to seamlessly integrate the TPL aware optimizations into early design stages.
To our best knowledge, this is the first work for TPL compliance at both standard cell and placement levels.
%including standard cell conflict removal, timing analysis, standard cell pre-coloring, and detailed placement.
An optimal graph model to simultaneously solve cell placement and color assignment is proposed,
and then a two-stage graph model is presented to achieve speedup.
Our framework is compared with traditional layout decomposition.
The results show that considering TPL constraints in early design stages can dramatically reduce the conflict number and stitch number in final layout.
As continuing growth of technology node to sub-16 nm, TPL turns out to be a definitely promising lithography solution.
A dedicated design flow that integrating TPL constraints is necessary to assist in the whole process.
We believe this paper will stimulate more research on TPL and TPL aware design.

\vspace{-.1in}
\section{Acknowledgment}

This work is supported in part by NSF grants CCF-0644316 and CCF-1218906, SRC task 2414.001, NSFC grant 61128010, and IBM Scholarship.